\documentclass[11pt,a4paper]{article}

\usepackage[utf8]{inputenc}
\usepackage[T1]{fontenc}
\usepackage{lmodern}
\usepackage{graphicx,amssymb,latexsym,amsmath,amsfonts,subfigure,epsf,bm,mathrsfs}
\usepackage{graphicx}
\usepackage{geometry}
\usepackage{hyperref}
\usepackage{natbib}

\newtheorem{proposition}{Proposition}
\newtheorem{proof}{Proof}
\newtheorem{remark}{Remark}
\newtheorem{theorem}{Theorem}
\newcommand{\eps}{\varepsilon}
\newcommand{\no}{\noindent}

\def\bfrho{\mbox{\boldmath$\varrho$}}

\newcommand{\ww}{\operatorname{\mathbf{w}}}

\title{Nonclassical Turing instabilities induced by superdiffusive transport in FitzHugh–Nagumo dynamics}

\author{
	Rossella Rizzo$^{1}\quad $ 
	Gaetana Gambino$^{2}\quad$
	Vincenzo Sciacca$^{2}\quad$ 
	Marco Sammartino$^{1}$
}

\date{}

\begin{document}
	
	\maketitle
	
	\begin{center}
		\small
		$^{1}$Department of Engineering, University of Palermo, Viale delle Scienze, Ed. 8, Palermo, 90128, Italy\\
		$^{2}$Department of Mathematics, University of Palermo, Via Archirafi 34, Palermo, 90123, Italy
	\end{center}
	
	\begin{abstract}
		We investigate diffusion-driven instabilities in a FitzHugh–Nagumo reaction–diffusion system with superdiffusive transport, modeled by fractional Laplacian operators with different diffusion orders for the activator and the inhibitor.
		
		A linear stability analysis yields explicit expressions for the instability threshold and the critical wavenumber and shows that superdiffusion modifies the band of unstable modes and the characteristic spatial scale of emerging patterns. We show that the threshold depends only on the ratio of the fractional exponents and on the kinetic parameters, while the spatial scale is controlled by the diffusion orders and the domain size.
		
		When the diffusion orders differ, spatial instabilities may occur even in regimes where the activator diffuses faster than the inhibitor, due to the combined effect of diffusion rates, anomalous scaling and system size. This leads to instability mechanisms that depart from the classical activator–inhibitor framework.
		
		A weakly nonlinear analysis near threshold provides the amplitude equation governing nonlinear saturation and reveals that superdiffusion promotes subcritical behavior. We also analyze the interaction between stationary and oscillatory instabilities near Turing–Hopf codimension-two points. All analytical results are supported by numerical simulations. 
			\end{abstract}
	
	\section{Introduction}
\label{intro}
In this work we investigate the effects of superdiffusive transport on diffusion--driven instabilities in a FitzHugh--Nagumo reaction--diffusion system, with particular emphasis on the role of different fractional diffusion orders for the activator and the inhibitor. Our goal is to assess how nonlocal transport mechanisms modify both the onset and the nonlinear development of spatial instabilities in a paradigmatic activator--inhibitor model.

The FitzHugh--Nagumo (FHN) model was originally introduced by FitzHugh \cite{FH1960} and subsequently validated by Nagumo \textit{et al.} \cite{NAY1962} as a two--variable reduction of the Hodgkin--Huxley equations for nerve membrane excitability. Due to its balance between mathematical simplicity and dynamical richness, the FHN system has become a canonical model for excitable media and has been widely used to describe a large class of phenomena, including neuronal and cardiac excitability, chemical reactions and other biological and physical systems \cite{CPRG2024}. Depending on parameter values, FHN--type equations can reproduce qualitatively different regimes such as excitability, oscillations and bistability \cite{MDBE1997,GLRS2024-1}. Moreover, FHN models and closely related activator--inhibitor systems have also been employed to describe prey--predator dynamics in ecological systems \cite{MPTML2002,BT2005,TB1994}.

In many complex or heterogeneous media, classical (Fickian) diffusion is often inadequate to describe transport processes, and anomalous diffusion, characterized by non--Gaussian statistics or long--range jumps, becomes relevant \cite{KBS1987,METZLER2000}. Fractional derivatives provide an effective and widely used framework to encode these nonlocal effects in space or time and fractional reaction--diffusion systems have therefore emerged as a powerful modeling tool for anomalous transport in biological tissues, porous materials and neural media \cite{METZLER2000,GMV2008,IWL2017,KZF2022}.

\subsection{Anomalous diffusion and fractional derivatives}

Classical diffusion is characterized by a linear growth of the mean square displacement \\
$\langle (\Delta \mathbf{x}(t))^2 \rangle \propto t$. In contrast, anomalous diffusion is characterized by the more general scaling:
\begin{equation}
	\langle (\Delta \mathbf{x}(t))^2 \rangle \propto t^\alpha, \qquad \alpha \in (0,2],
\end{equation}
where $\alpha$ is the anomalous diffusion exponent. We speak of subdiffusion for $0<\alpha<1$ and of superdiffusion for $1<\alpha<2$ \cite{KBS1987,METZLER2000,GMV2008,ALVES2016,IWL2017}. For $\alpha=1$ one recovers normal diffusion, while $\alpha=2$ corresponds to ballistic motion. Both regimes play an important role in physical, chemical, biological, and geological systems. Subdiffusion typically arises in systems with long waiting times, such as transport in porous media or crowded biological environments \cite{AMYPL1996,DZ1999,WEISS2003}, whereas superdiffusion is associated with long--range jumps and is observed, for instance, in plasmas, turbulent flows and charge carrier transport in semiconductors \cite{SQS1993,CLZ2001,TONER2005,ZYM2009,MARSHALL2014}. A paradigmatic example of superdiffusion is provided by L\'evy flights \cite{METZLER2000,MK2004,uchaikin2013}.

At the macroscopic level, classical diffusion is described by the standard Laplace operator $\Delta=\partial_{xx}+\partial_{yy}+\partial_{zz}$. 
On the other hand, L\'evy flights of a substance $u$ are described by the fractional Laplacian $\Delta^{\mu/2}\equiv \nabla^\mu$, defined in Fourier space as:
\begin{equation}
	\Delta^{\mu/2} u = \mathcal{F}^{-1}\!\left[-|\mathbf{k}|^\mu \, \mathcal{F}[u]\right],
\end{equation}
where $\mathcal{F}$ denotes the Fourier transform and $\mathcal{F}^{-1}$ its inverse and $0<\mu\le 2$ \cite{METZLER2000,GMV2008,IWL2017,KZF2022}.

A particular class of superdiffusive processes, such as L\'evy flights, corresponds to random walks whose jump--length distribution has infinite second moment, so that the mean squared displacement is formally divergent. Nevertheless, by introducing suitable cut--offs or scaling arguments based on fractional moments, one can define a pseudo mean squared displacement that characterizes the typical spread of the process \cite{BOUCHAUD1990,SZK1993,FOGEDBY1994,KSZ1996,METZLER2000}. 
In this work we focus on the superdiffusive regime $1<\mu\le 2$, for which the typical displacement scales as $|\mathbf{x}(t)|_{\mathrm{typ}}\sim t^{1/\mu}$ and the pseudo mean squared displacement scales as $t^{2/\mu}$. Since the jump--length distribution obeys $P(|Y|>r)\sim r^{-\mu}$, smaller values of $\mu$ correspond to heavier tails and a higher probability of long jumps.

\subsection{Superdiffusion in reaction--diffusion systems}

The introduction of anomalous or superdiffusive transport in reaction--diffusion systems has attracted considerable attention over the last two decades, motivated by the need to describe pattern--forming systems embedded in complex or heterogeneous environments where long--range transport processes cannot be neglected \cite{GD2008,GAFIYCHUK2010,GMV2008,ZT2014,IWL2017,KZF2022}. In this context, fractional diffusion operators provide a natural framework to model nonlocal transport mechanisms in spatially extended systems.

The presence of superdiffusive transport modifies both the linear stability properties and the nonlinear pattern selection mechanisms of classical reaction–diffusion models. In particular, since fractional diffusion alters the balance between short- and long-range interactions, it can significantly affect the standard activator–inhibitor paradigm underlying diffusion-driven instabilities.

Several works have already explored these ideas in specific settings. Golovin \textit{et al.} \cite{GMV2008} analyzed a fractional Brusselator model and showed that diffusion--driven instabilities may occur even when the inhibitor diffuses slower than the activator, provided it has a higher diffusion exponent. Zhang and Tian \cite{ZT2014} and Iqbal \textit{et al.} \cite{IWL2017} studied fractional FitzHugh--Nagumo systems assuming identical diffusion exponents for the two species, mainly by numerical means. Khudhair \textit{et al.} \cite{KZF2022} investigated pattern formation in fractional Schnakenberg systems under the same assumption of equal diffusion orders. Numerical continuation studies have further shown that fractional diffusion can modify bifurcation structures in nonlinear PDEs \cite{EKS2021}

Most of the scientific interest in Lévy-type diffusion stems from its ability to capture effective long-range transport in complex systems. Representative examples include neural communication in large-scale brain networks, characterized by highly nonlocal synaptic connectivity and preferential displacements driven by foraging strategies or predator–prey interactions in ecological systems \cite{VABMPS1996,VBHLRS1999,SRHSWMT2014}.

Although action potentials propagate locally along axons, neural networks contain long-range synaptic connections, so that the effective propagation of activity across neural populations may exhibit features analogous to long-jump processes.

Lévy walks have also been widely hypothesized to represent optimal search strategies when resources are sparse or randomly distributed, a hypothesis supported by empirical studies on bees, albatrosses, deer \cite{VABMPS1996,VBHLRS1999}, and marine predators \cite{LKW1988,SSH2008,SRHSWMT2014}. From an evolutionary perspective, Lévy-like strategies provide higher encounter rates in prey-scarce environments, whereas Brownian motion is more efficient in prey-rich conditions \cite{VBHLRS1999,SSH2008}.

Levandowsky et al. \cite{LKW1988} proposed Lévy-distributed steps as an optimal strategy for microzooplankton foraging in sparse environments, while Viswanathan et al. \cite{VABMPS1996,VBHLRS1999} and Sims et al. \cite{SSH2008} reported power-law flight-length distributions and Lévy-walk-like behavior across insects, mammals, birds, and marine predators. More recently, fossil trail analyses revealed signatures of Lévy-type search patterns in extinct marine species \cite{SRHSWMT2014}.

Despite these advances, a systematic analytical understanding of diffusion--driven instabilities in FitzHugh--Nagumo systems with {different} superdiffusion orders for the activator and the inhibitor is still lacking. In particular, explicit analytical results for the instability threshold, the critical wavenumber and the nonlinear saturation of the patterns are still missing. Also the extent to which anomalous transport can modify the classical activator--inhibitor paradigm remains only partially explored.

\subsection{Outline and main results}

In this work we analyze diffusion-driven instabilities in a FitzHugh--Nagumo system with fractional transport and heterogeneous diffusion orders.

We first derive explicit analytical expressions for the instability threshold and the critical wavenumber. We show that, when the diffusion orders differ, the threshold depends only on their ratio and on the kinetic parameters, while being independent of their individual values. This property provides a natural framework for classifying instability regimes in anomalous reaction--diffusion systems.

We then demonstrate that fractional transport allows for nonclassical instability mechanisms. In particular, spatial instabilities may arise even in regimes where the activator diffuses faster than the inhibitor, as a consequence of the interplay between diffusion rates, anomalous scaling, and domain size. This behavior departs from the classical short-range activation and long-range inhibition paradigm.

Next, we perform a weakly nonlinear analysis near the instability threshold and derive the corresponding amplitude equation. We show that superdiffusion strongly influences nonlinear saturation and systematically promotes subcritical behavior, thereby enhancing finite-amplitude effects.

Finally, we analyze the interaction between stationary and oscillatory instabilities near Turing--Hopf codimension--two points and characterize the resulting spatio-temporal dynamics.

All analytical results are supported by numerical simulations.

The paper is organized as follows. In Sec.~\ref{sec_model} we introduce the fractional FitzHugh--Nagumo  system that will be studied throughout the paper and fix the notation. In Sec.~\ref{linearinstabilityanalysis} we analyze the homogeneous FitzHugh--Nagumo system and classify the different dynamical regimes in parameter space, and we then study the linear stability of the spatially extended system and derive the conditions for the onset of diffusion--driven instabilities. In Sec.~\ref{WNL} we perform a weakly nonlinear analysis near the Turing threshold and derive the corresponding amplitude equation. In Sec.~\ref{TH} we study the dynamics near codimension--two Turing--Hopf bifurcation points. Finally, Sec.~\ref{conclusions} is devoted to a discussion of the results and to perspectives for future work.

\section{The fractional FHN model}\label{sec_model}

We consider a generalized FHN system coupled with superdiffusive transport in a one--dimensional spatial domain $\Omega=[0,\ell]$:
\begin{align}
	\frac{\partial U}{\partial \tau} &= U(B - AU^2) - CV + EUV + D_U \, \partial^{\alpha_1}_{\tilde x} U, \label{sys_dim_1} \\
	\frac{\partial V}{\partial \tau} &= FU - GV - H + D_V \, \partial^{\alpha_2}_{\tilde x} V. \label{sys_dim_2}
\end{align}
The system is considered together with prescribed initial conditions and periodic boundary conditions.

To derive the dimensionless form of the system, we introduce the following nondimensional variables:
\begin{equation}
	\label{scales}
	\tilde{t}=\frac{\ell^{\alpha_1}}{D_U}  {\tau} ,\quad \tilde{x}=\ell x,\quad \tilde{u}=\sqrt{\frac{B}{A}},\quad \tilde{v}= \frac{B}{C}\sqrt{\frac{B}{A}},
\end{equation}
which lead to the following dimensionless formulation:
\begin{align}
	\frac{\partial u}{\partial t} &= \Gamma(u(1-u^2)-(1-\beta u)v)+\partial^{\alpha_1}_{x}u, \label{origsyst_1} \\
	\frac{\partial v}{\partial t} &= \Gamma(\varepsilon(\gamma u-v-a))+d\partial^{\alpha_2}_{x}v. \label{origsyst_2}
\end{align}
The dimensionless parameters are defined as
\begin{equation}
	\label{parameter_adim}
	\begin{split}
		&\Gamma=\frac{B\ell^{\alpha_1}}{D_U},\quad \beta=\frac{E}{C}\sqrt{\frac{B}{A}},\quad\varepsilon=\frac{G}{B},\\
		&\gamma=\frac{FC}{BG},\quad a=\frac{HC}{BG}\sqrt{\frac{A}{B}},\quad d=\frac{D_V}{D_U}\ell^{\alpha_1-\alpha_2}.
	\end{split}
\end{equation}

In system \eqref{origsyst_1}--\eqref{origsyst_2}, $u(x,t)$ and $v(x,t)$ represent the activator and inhibitor species, respectively, and the constants $\alpha_1$ and $\alpha_2$ denote the fractional diffusion exponents associated with each species. In this work we focus on the superdiffusive regime and assume
\begin{equation}
	\label{gamma1-2cond}
	1<\alpha_i\leq 2, \qquad i=1,2.
\end{equation}
When $\alpha_1=\alpha_2=2$, the system reduces to the classical FHN model with standard diffusion \cite{GLRS2019}. If $\alpha_1=\alpha_2$, the parameter $d$ represents the ratio of the diffusion coefficients of $v$ and $u$. In the more general case $\alpha_1\neq\alpha_2$, the parameter $d$ also depends on the size of the spatial domain due to the different scaling behavior of the diffusion operators.

The parameter $\varepsilon$ denotes the ratio of the characteristic timescales of the two species, while $\gamma$ and $a$ control the number and position of the nullcline intersections. The parameter $\beta\in\mathbb{R}$ breaks the symmetry of the system under the transformation $(u\rightarrow -u, v\rightarrow -v, a\rightarrow -a)$, see \cite{DBBM2001,BMBD2001,BDWDB2002,TSB2012}. We choose $|\beta|<1$ so that the typical shape of the classical FHN $u$--nullcline is preserved and not significantly affected by the presence of the vertical asymptote $u=1/\beta$ in the phase plane $(u,v)$. Under this condition, if $\beta>0$ the asymptote lies to the right of the cubic, while if $\beta<0$ it lies to the left.

In the rest of the paper, and in accordance with previous literature \cite{DBBM2001,BMBD2001,BDWDB2002,TSB2012}, we restrict to $0\leq\beta<1$. The analysis for $-1<\beta\leq 0$ is analogous and leads to similar results. The parameter $\Gamma>0$ quantifies the strength of the reaction terms and depends on the domain size. All other kinetic parameters are nonnegative, except $a$, which may take both positive and negative values. Therefore, throughout the paper we assume
\begin{equation}
	\label{cond_parameters}
	a\in\mathbb{R}, \quad 0\leq \beta<1, \quad \gamma>0, \quad \varepsilon>0, \quad \Gamma>0, \quad d>0.
\end{equation}

Finally, we introduce the ratio of the two fractional exponents
\begin{equation}
	\label{gammacond}
	\alpha=\frac{\alpha_1}{\alpha_2}, \qquad \text{so that} \qquad \frac{1}{2}<\alpha<2.
\end{equation}

\section{Linear stability analysis}
\label{linearinstabilityanalysis}

\setcounter{equation}{0}

In this Section, we first review the conditions for the existence and stability of a unique steady state $E^*$ of the FHN system in the absence of diffusion. We then derive the conditions for the onset of a superdiffusion-driven Turing bifurcation.

\subsection{Local dynamics}

Let us consider the system \eqref{origsyst_1}–\eqref{origsyst_2} without spatial effects:
\begin{align}
	\frac{d u}{d t}&=\Gamma(u(1-u^2)-(1-\beta u)v),\label{kin_1}\\
	\frac{d v}{d t}&=\Gamma(\varepsilon(\gamma u-v-a)). \label{kin_2}
\end{align}
Depending on the values of the parameters $\beta$, $\gamma$, and $a$, the nullclines of the system \eqref{kin_1}–\eqref{kin_2} may intersect in different ways, leading to distinct dynamical regimes \cite{HM1994,GGLR2022,CPRG2024}. In particular, the system may exhibit the following mutually exclusive cases: (i) a unique stable equilibrium located on the inner branch of the $u$--nullcline (monostable case), (ii) a unique stable equilibrium on one of the outer branches (excitable case), or (iii) three equilibria, two stable and one unstable (bistable case).

In the present work we restrict our analysis to parameter regimes for which the system admits a {unique} equilibrium point. This choice is motivated by the fact that both the Turing instability analysis and the subsequent weakly nonlinear theory are local bifurcation theories formulated in the neighborhood of a single homogeneous steady state. In bistable regimes, where two stable equilibria coexist, diffusion--driven instabilities could in principle occur around more than one equilibrium, leading to mode interactions and resonance phenomena and possibly to the emergence of more complex structures such as superlattice patterns \cite{DBBM2001,TSB2012}. The analysis of such scenarios requires a different theoretical framework and is beyond the scope of the present paper.

In the following, we summarize the conditions on the parameters under which the system \eqref{kin_1}–\eqref{kin_2} admits a unique equilibrium and falls into cases (i) or (ii). First, in Proposition~\ref{propexistencemonoeq}, we provide the conditions ensuring the existence of a unique equilibrium point $E^* = (u^*, v^*)$. We assume that
\begin{equation}\label{us_on_the_left}
	1-\beta u^*> 0,
\end{equation}
so that the equilibrium point lies to the left of the vertical asymptote $u=1/\beta$ of the $u$--nullcline in the phase plane $(u,v)$.

\begin{proposition}[Existence of a unique equilibrium] \label{propexistencemonoeq}
	Consider the system \eqref{kin_1}–\eqref{kin_2} and define
	\begin{equation}\label{c_m1} 
		p := a\beta + \gamma - 1 - \frac{\beta^2 \gamma^2}{3}, 
	\end{equation}  
	\begin{equation} 
		q := \frac{\beta\gamma(a\beta + \gamma - 1)}{3} - \frac{2\beta^3 \gamma^3}{27}-a. 
	\end{equation}
	The system admits a unique equilibrium point $E^* = (u^*, v^*)$, where $v^* = \gamma u^* - a$, if and only if
	\begin{equation}
		\label{cond_uniqueeq}
		4p^3+27q^2>0.
	\end{equation}
\end{proposition}

For the proof of Proposition~\ref{propexistencemonoeq}, we refer the reader to the Appendix of \cite{GGLR2022}.

To analyze the stability of the steady state, we linearize the system \eqref{kin_1}–\eqref{kin_2} around $E^*$:
\begin{equation}\label{lin_kin}
	\dot{\mathbf{w}}=J \mathbf{w}, \quad {\rm with}\quad
	\mathbf{w}=\left(\begin{array}{c}
		u-u^*\\v-v^*
	\end{array}\right),
\end{equation}
\begin{equation}  
	J = \left(\begin{array}{cc}
		\varepsilon_H &  -(1-\beta u^*)\\
		\eps\gamma & -\eps
	\end{array}\right),
	\label{J}
\end{equation}
where
\begin{equation}\label{epsH} 
	\varepsilon_H := 1 + \beta v^* - 3u^{*2}. 
\end{equation}
This yields the characteristic equation
\begin{equation}\label{disp_cin} 
	\lambda^2 - \Gamma \mathrm{tr}(J) \lambda + \Gamma^2 \mathrm{det}(J) = 0, 
\end{equation}
with
\begin{align}
	\rm{tr}(J)&=-\varepsilon+\varepsilon_H, \label{linearstabilityconditions1}\\
	\det(J)&=\varepsilon \left[-\varepsilon_H+\left(1-\beta u^*\right)\gamma \right].
	\label{linearstabilityconditions2}
\end{align}
Stability requires $\mathrm{tr}(J) < 0$ and $\mathrm{det}(J) > 0$.

Let $v(u)= u(1-u^2)/(1-\beta u)$ denote the $u$--nullcline. The two extrema of this cubic--like curve divide it into three branches: the {inner branch}, located between the two extrema, and the two {outer branches}, located on either side. Since $\varepsilon_H = v'(u^*)/(1-\beta u^*)$ and condition \eqref{us_on_the_left} holds, the sign of $\varepsilon_H$ coincides with the sign of $v'(u^*)$. Therefore, $\varepsilon_H>0$ corresponds to an equilibrium located on the inner branch of the $u$--nullcline, while $\varepsilon_H<0$ corresponds to an equilibrium located on one of the outer branches.

Building upon this, in Propositions \ref{propkinstabmono}, \ref{prophopf}, and \ref{propkinstabexci}, we identify the parameter regimes for which the unique equilibrium $E^*$ corresponds to a monostable configuration, undergoes a Hopf bifurcation, or belongs to the excitable regime.

\begin{proposition}{(Monostable case)}\label{propkinstabmono}
	Let the assumptions of Proposition~\ref{propexistencemonoeq} hold. If
	\begin{equation}\label{stab_mono} 
		0 < \varepsilon_H < \min\left( \varepsilon, \gamma(1 - \beta u^*) \right),
	\end{equation}
	then the unique equilibrium point $E^* = (u^*, v^*)$ lies on the inner branch of the $u$--nullcline and is linearly stable.
\end{proposition}
Indeed, since $\varepsilon_H>0$, the equilibrium $E^*$ lies on the inner branch of the $u$--nullcline. Moreover, under condition \eqref{stab_mono} one has $\mathrm{tr}(J)<0$, see \eqref{linearstabilityconditions1}, and $\det(J)>0$, see \eqref{linearstabilityconditions2}, so that $E^*$ is linearly stable.

\begin{proposition}{(Hopf bifurcation)}\label{prophopf}
	Assume that the hypotheses of Proposition~\ref{propexistencemonoeq} hold and that the equilibrium $E^*$ lies on the inner branch of the $u$--nullcline. If
	\[
	\varepsilon_H = \varepsilon < \gamma(1-\beta u^*),
	\]
	then $E^*$ undergoes a Hopf bifurcation.
\end{proposition}
In this case $\mathrm{tr}(J)=0$ and $\det(J)>0$, so that a pair of purely imaginary eigenvalues crosses the imaginary axis and the equilibrium $E^*$ undergoes a Hopf bifurcation.

\begin{proposition}{(Excitable case)}\label{propkinstabexci}
	Assume that the hypotheses of Proposition~\ref{propexistencemonoeq} hold. If
	\[
	\varepsilon_H < 0,
	\]
	then the unique equilibrium point $E^* = (u^*, v^*)$ lies on one of the outer branches of the $u$--nullcline and is linearly stable.
\end{proposition}
Indeed, if $\varepsilon_H<0$, the equilibrium $E^*$ lies on one of the outer branches of the $u$--nullcline. Moreover, one has $\mathrm{tr}(J)<0$ and $\det(J)>0$, which implies linear stability.

In Fig.~\ref{cond_regimes} we illustrate the different dynamical regimes of the local FHN system \eqref{kin_1}–\eqref{kin_2} in the $(a,\gamma)$ parameter plane, namely the monostable (black region), excitable (dark grey region), and bistable regimes (light grey region). The figure also highlights the role of the parameters $\varepsilon$ and $\beta$: decreasing $\varepsilon$ may destabilize the unique equilibrium through a Hopf bifurcation as predicted by Proposition~\ref{prophopf}, see the white area in Fig.~\ref{cond_regimes}(b). On the other hand, a nonzero value of $\beta$ breaks the symmetry of the system and renders the boundaries between the different regimes asymmetric with respect to the line $a=0$, see Fig.~\ref{cond_regimes}(c).

\begin{figure*}
	\centering
	\subfigure[]{\includegraphics[width=0.32\textwidth]{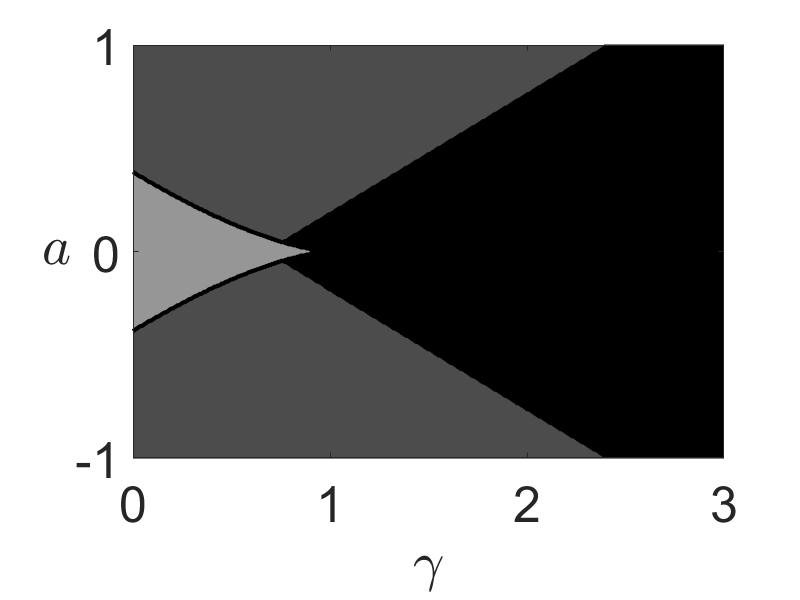}
		\label{beta=0_varepsilon=2}}
	\subfigure[]{\includegraphics[width=0.32\textwidth]{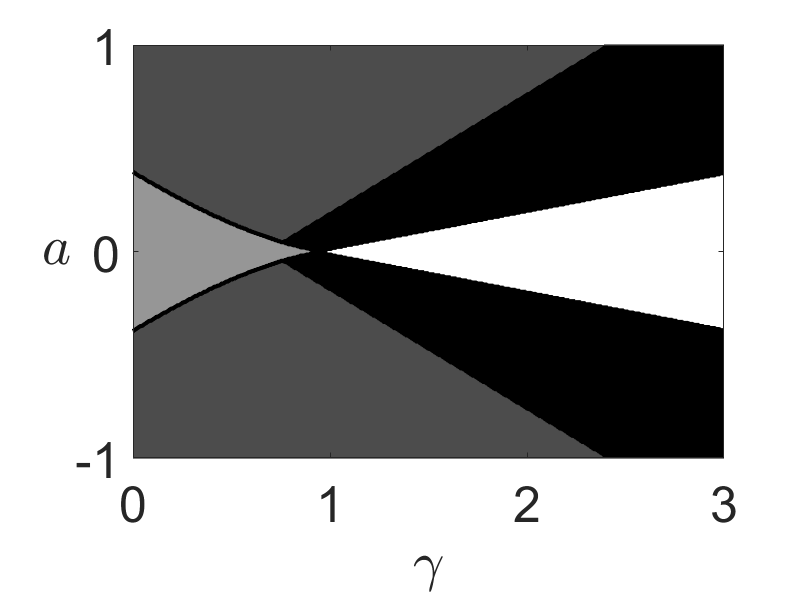}\label{beta=0_varepsilon=0_9}}
	\caption{Different dynamical regimes of the local FitzHugh--Nagumo system \eqref{kin_1}–\eqref{kin_2} in the $(a,\gamma)$ plane. The black, dark gray, and light gray regions correspond to the monostable, excitable, and bistable regimes, respectively. The white region indicates the parameter values for which the unique equilibrium loses stability through a Hopf bifurcation. (a) $\beta=0$, $\varepsilon=2$. (b) $\beta=0$, $\varepsilon=0.9$.} 
	\label{cond_regimes}
\end{figure*}

\subsection{Diffusion driven instability}

In this Section, by means of a linear stability analysis, we investigate under which conditions the equilibrium point $E^*$, which is linearly stable in the absence of diffusion, may lose its stability in the presence of spatial perturbations through a Turing bifurcation. Specifically, we will show that in the excitable case, i.e. under the assumptions of Proposition~\ref{propkinstabexci}, such a bifurcation cannot occur.
In the monostable case, corresponding to the assumptions of Proposition~\ref{propkinstabmono}, we will show that $E^*$ may lose stability through a Turing mechanism. We will determine the bifurcation threshold $d_c$ and the corresponding critical wavenumber $k_c$.

To investigate the linear stability of the equilibrium $E^*$, we substitute the normal mode ansatz $\mathbf{w}=\mathbf{W}_k e^{\sigma t + i k x}$ into \eqref{origsyst_1}--\eqref{origsyst_2}, where $\sigma$ is the growth rate of the perturbation, $k$ is the associated wavenumber, and $\mathbf{w}$ is defined in \eqref{lin_kin}. This yields the following eigenvalue problem:

\begin{equation}\label{lin_sys}
	\sigma \mathbf{W}_k=\mathcal{L}\mathbf{W}_k, \quad {\rm with}\quad
	\mathcal{L}=\Gamma J-D_k,
\end{equation}

\begin{equation} \quad {\rm and}  \quad D_k=\left(\begin{array}{cc}
		k^{\alpha_1}&0\\0&dk^{\alpha_2}
	\end{array}\right).
	\label{vectsyst}
\end{equation}

\no where $J$ is given in \eqref{J}. The eigenvalue $\sigma$, regarded as a function of the wavenumber $k$, satisfies the following dispersion relation:

\begin{equation}
	\sigma^2+g\left(k\right)\sigma + h\left(k\right)=0,
	\label{dispersionrelation}
\end{equation}
where:
\begin{eqnarray}
	g\left(k\right)&=&-{\rm tr}(\Gamma J-D_k)\label{g(k)}\\
	\ &=&k^{\alpha_1}+dk^{\alpha_2}+\Gamma(\varepsilon-\varepsilon_H),\nonumber\\
	h\left(k\right)&=&\det(\Gamma J-D_k)\label{h(k)}\\
	\ &=&(k^{\alpha_1}-\Gamma \varepsilon_H)(dk^{\alpha_2}+\Gamma\varepsilon)+\Gamma^2\varepsilon\gamma(1-\beta u^*).\nonumber
\end{eqnarray}
We recall that a Turing instability of the homogeneous equilibrium $E^*$ occurs when, for some mode $k \neq 0$, the corresponding eigenvalue $\sigma(k)$ is real and crosses zero from negative to positive values.

We can now state the following Theorems~\ref{propturininstabexci} and~\ref{propturininstabmono}.

\begin{theorem}{\bf{(No Turing instability in the excitable regime)}}\label{propturininstabexci}
	Consider the system \eqref{origsyst_1}–\eqref{origsyst_2}, and suppose that the assumptions of Propositions \ref{propexistencemonoeq} and \ref{propkinstabexci} are satisfied so that the equilibrium point $E^*$, located on an outer branch of the $u$-nullcline, is linearly stable for the system \eqref{kin_1}-\eqref{kin_2}. Then, the equilibrium $E^*$ of the system \eqref{origsyst_1}–\eqref{origsyst_2} cannot undergo a Turing bifurcation.
\end{theorem}

\begin{proof}
	Since the assumptions of Proposition \ref{propkinstabexci} hold, we have $\varepsilon_H < 0$. As a result, both $g(k)$ and $h(k)$, defined in \eqref{g(k)} and \eqref{h(k)}, are positive for all $k \in \mathbb{R}^+$. By Descartes' rule of signs, the roots of the dispersion relation \eqref{dispersionrelation} must have negative real parts for all $k \in \mathbb{R}^+$, and thus no instability can occur.
\end{proof}

\begin{theorem}{\bf{(Turing instability in the monostable regime)}}\label{propturininstabmono} 
	Consider the system \eqref{origsyst_1}–\eqref{origsyst_2}, and suppose that the assumptions of Propositions \ref{propexistencemonoeq} and \ref{propkinstabmono} are satisfied so that the equilibrium point $E^*$, located on the inner branch of the $u$-nullcline, is linearly stable for the system \eqref{kin_1}-\eqref{kin_2}. Then, the steady state $E^*$ of the system \eqref{origsyst_1}–\eqref{origsyst_2} exhibits a Turing bifurcation at $d=d_c$, where:

		\begin{equation}\label{dc}
			d_c=\frac{\varepsilon\Gamma\left((1-\beta u^*)\gamma(1-\alpha)+\sqrt{\delta}\right)}{\left({\Gamma}\left(\varepsilon_H-\frac{\gamma}{2}(1+\alpha)(1-\beta u^*)+\frac{\sqrt{\delta}}{2}\right)\right)^{1/{\alpha}}\left(\gamma(1+\alpha)(1-\beta u^*)-\sqrt{\delta}\right)}, 
		\end{equation}
		and: 
		\begin{equation}\label{delta}
			\delta=(\gamma(1+\alpha)(1-\beta u^*)-2\varepsilon_H)^2+4\varepsilon_H(\gamma(1-\beta u^*)-\varepsilon_H).
		\end{equation}

	At $d=d_c$ the critical wavenumber is:
	
	\begin{equation}\label{kc}
		k_c=\left({\Gamma}\left(\varepsilon_H-\frac{\gamma}{2}(1+\alpha)(1-\beta u^*)+\frac{\sqrt{\delta}}{2}\right)\right)^{1/{\alpha_1}}.
	\end{equation}
\end{theorem}

\begin{proof}
	We seek a wavenumber $k \in \mathbb{R}^+$ such that the corresponding eigenvalue $\sigma(k)$ is real and changes sign from negative to positive.
	
	Since the hypotheses of Proposition~\ref{propkinstabmono} hold, it follows that $g(k) > 0$ for all $k \in \mathbb{R}^+$. Hence, for the dispersion relation~\eqref{dispersionrelation} to admit a positive root for some $k \in \mathbb{R}^+$, it is necessary that $h(k) < 0$, which also ensures that the corresponding eigenvalue $\sigma(k)$ is real. From the expression of $h(k)$ given in~\eqref{h(k)}, we derive the following necessary condition for instability:
	\begin{equation}
		\label{neccondGamma}
		k^{\alpha_1}-\Gamma\varepsilon_H<0.
	\end{equation}
	To simplify the analysis, we introduce the change of variables:
	\begin{equation}
		\label{k_gamma1}
		z=k^{\alpha_2}\quad \Rightarrow \quad z^\alpha=k^{\alpha_1}.
	\end{equation}
	so that the function $h(k)$ becomes:
	\begin{equation}\label{h(z)}
		h(z)=(z^\alpha-\Gamma\varepsilon_H)(\Gamma \varepsilon+dz)+\Gamma^2\varepsilon\gamma(1-\beta u^*)
	\end{equation}
	where we now consider $z\in \mathbb{R}^+$.	
	
	We impose the marginal instability condition by requiring the minimum of $h(z)$ over $z>0$ be zero, which is equivalent to finding a critical point $z^* \in \mathbb{R}^+$ satisfying the following conditions:
	\begin{subequations}\label{marg}
		\begin{align}
			\begin{array}{l}
				h(z^*) = 0, 
			\end{array}
			& \label{marg:a} \\[1ex]
			\begin{array}{l}
				h^{'}(z^*) = 0, 
			\end{array}
			&  \label{marg:b} \\[1ex]
			\begin{array}{l}
				h^{''}(z^*) > 0,
			\end{array}
			& \label{marg:c}
		\end{align}
	\end{subequations}
	where $'=\frac{d}{dz}$, and: 
	\begin{align}\label{h_z}
		h'&=(\alpha+1)dz^\alpha+\Gamma\alpha\varepsilon z^{\alpha-1}-\Gamma d \varepsilon_H,\\
		\label{h_zz}
		h^{''}&=\alpha z^{\alpha-2}\left((\alpha+1)dz+\Gamma(\alpha-1)\varepsilon\right).
	\end{align}
	Assume such a point $z^*>0$ exists. Then, the corresponding critical value of the bifurcation parameter $d$ can be computed from~\eqref{marg:a}: 
	\begin{equation}\label{dbar}
		d=d_c=\frac{-\Gamma\varepsilon\left(z^{*\alpha}+\Gamma(-\varepsilon_H+\gamma(1-\beta u^*))\right)}{z^*(z^{*\alpha}-\Gamma\varepsilon_H)}.
	\end{equation} 
	Evaluating the derivative $h'(z)$ in $z=z^*$
	and substituting the expression for $d$ from~\eqref{dbar} into~\eqref{h_z}, we obtain:
	\begin{equation}\label{h_htilde}
		h^{'}(z^*)=\frac{\Gamma\varepsilon}{z^*(\Gamma\varepsilon_H-z^{*\alpha})}\tilde{h}(z^{*\alpha}),
	\end{equation}
	where:
	\begin{equation}\label{htilde}
		\begin{split}
			\tilde{h}(z^{*\alpha}) &= (z^{*\alpha})^2+\Gamma\big(\gamma(1+\alpha)(1-\beta u^*)-2\varepsilon_H \big) z^{*\alpha}\\&-\Gamma^2\varepsilon_H\big(\gamma(1-\beta u^*)-\varepsilon_H \big).
		\end{split}
	\end{equation}
	The constant term in $\tilde{h}(z^{*\alpha})$ in \eqref{htilde} is negative by assumption (as ensured by Proposition~\ref{propkinstabmono}). Therefore, according to Descartes' rule of signs, 
	$\tilde{h}(z^{*\alpha})$ has exactly one real positive root, regardless of the sign of the linear coefficient. This positive root is given by: 
	\begin{equation}\label{z_piu}
		z^{*\alpha}=\frac{\Gamma}{2}(2\varepsilon_H-\gamma(1+\alpha)(1-\beta u^*)+\sqrt{\delta}),
	\end{equation}
	where $\delta$ is defined in~\eqref{delta}.
	
	From~\eqref{h_htilde}, it follows that
	\begin{equation}\label{z_star}
		z^*=\left(\frac{\Gamma}{2}(2\varepsilon_H-\gamma(1+\alpha)(1-\beta u^*)+\sqrt{\delta})\right)^{\frac{1}{\alpha}},
	\end{equation}
	satisfies $h'(z^{*})=0$, thus verifying condition~\eqref{marg:b}. Substituting \eqref{z_star} into~\eqref{dbar} yields the critical value $d=d_c$, as given in \eqref{dc} for which condition~\eqref{marg:a} is also satisfied. It remains to verify that the value $z=z^*$
	in \eqref{z_star} also satisfies condition \eqref{marg:c}. To this end, substituting $d=d_c$, as given in \eqref{dc}, and $z=z^*$, as given in \eqref{z_star}, into \eqref{marg:c}, we obtain:
	\begin{equation}\label{der_sec} 
		h^{''}(z^*)=\frac{2\alpha\varepsilon\Gamma\sqrt{\delta}z^{*(\alpha-2)}}{\gamma(1+\alpha)(1-\beta u^*) - \sqrt{\delta}}, 
	\end{equation} 
	and it is straightforward to verify that the expression in \eqref{der_sec} is always positive. We can thus conclude that $z = z^*$ is the desired critical value that satisfies conditions \eqref{marg}. By reverting the change of variables~\eqref{k_gamma1}, we obtain the corresponding critical wavenumber $k = k_c$, as given in~\eqref{kc}.
	
	Finally, it is easy to verify that the expression for $d = d_c$ given in \eqref{dc} is positive. Indeed, using \eqref{kc} we can rewrite $d_c$ as follows:
	\begin{equation}\label{dc_n}
		d_c=\frac{\varepsilon\Gamma\left((1-\beta u^*)\gamma(1-\alpha)+\sqrt{\delta}\right)}{k_c^{\alpha_2}\left(\gamma(1+\alpha)(1-\beta u^*)-\sqrt{\delta}\right)}, 
	\end{equation}
	The denominator is always positive, as it is the product of 
	$k_c$, which is positive, and a term that can be readily shown to be positive. Similarly, the numerator can be shown to be positive under the assumptions of Proposition \ref{propkinstabmono}.
\end{proof}

\begin{remark}\label{rem1}
	Observe that the bifurcation threshold $d_c$ in \eqref{dc} does not depend on the individual diffusion powers $\alpha_1$ and $\alpha_2$, but only on their ratio $\alpha$ and on the kinetic parameters.
\end{remark}

This result highlights that the Turing threshold is governed by the relative scaling of the two transport processes, rather than by their absolute superdiffusive exponents.

\begin{remark}[Dependence of $k_c$ on $\Gamma$]\label{rem2}
	Observe that the necessary condition \eqref{neccondGamma} is always satisfied under the hypotheses of Proposition~\ref{propkinstabmono} and condition~\eqref{us_on_the_left}.
	
	Notice that the condition \eqref{neccondGamma} provides an upper bound for the critical wavenumber, which depends on $\Gamma$.

	Theorem \ref{propturininstabmono} gives the necessary conditions for the system \eqref{origsyst_1}-\eqref{origsyst_2} to admit
	a finite $k$ pattern-forming Turing instability in the case $d>d_c$, but it does not guarantee the emergence of spatial patterns.
	The pattern in fact will emerge only if 
	there is at least one integer $n$ so that $k_1<n<k_2<(\Gamma \varepsilon_H)^{1/\alpha_1}$,
	where 
	$k_1$ and $k_2$ are the roots of 
	$h(k)$ in \eqref{h(k)} and  $n \, (n\in \mathbb{N})$
	are the modes allowed by the periodic
	boundary conditions on the interval
	$[0,2\pi] \subseteq \mathbb{R}$.
	Therefore, to observe pattern formation under periodic boundary conditions, $\Gamma$ must be chosen sufficiently large.			\end{remark}

This reflects the fact that larger spatial domains (i.e. larger values of $\Gamma$) allow for a richer spectrum of unstable modes and therefore favor the emergence of spatial patterns.

\section{Effects of anomalous diffusion on pattern formation}

\subsection{Case $\alpha_1=\alpha_2$}

We first consider the case in which the two species superdiffuse with the same fractional order, i.e.\ $\alpha_1=\alpha_2$, or equivalently $\alpha=1$. Under this condition, both the activator and the inhibitor are characterized by the same transport scaling, and anomalous diffusion affects only the spatial redistribution of each species, without introducing any relative asymmetry between them. Recall that smaller values of $\alpha_1=\alpha_2$ correspond to stronger superdiffusion, namely to a higher probability of long jumps in the underlying transport process.

\paragraph{Bifurcation threshold.}
When $\alpha_1=\alpha_2$, the expression of the Turing threshold $d_c$ in \eqref{dc} reduces to:
\begin{equation}
	\label{dc_alpha=1}
	d_c=\frac{\varepsilon }{\left[\sqrt{\gamma(1-\beta u^*)} - \sqrt{\gamma(1-\beta u^*) - \varepsilon_H} \right]^2},
\end{equation}
which coincides with the classical threshold obtained in the presence of standard diffusion. In particular, $d_c$ does not depend on the width parameter $\Gamma$ nor on the common diffusion exponent $\alpha_1=\alpha_2$. Under the hypotheses of Proposition~\ref{propkinstabmono} and condition~\eqref{us_on_the_left}, it is straightforward to verify from the expression of $d_c$ in~\eqref{dc_alpha=1} that $d_c>1$. Moreover, since in this case the scaling \eqref{parameter_adim} reduces to $d=D_V/D_U$, instability can only arise through the classical Turing mechanism, namely when the inhibitor diffuses faster than the activator. Therefore, when $\alpha_1=\alpha_2$, anomalous diffusion does not modify the mechanism responsible for the onset of patterns, but only affects their spatial scales.

\paragraph{Critical wavenumber and wavelength selection.}
In the case $\alpha_1=\alpha_2$, the expression of the critical wavenumber given in \eqref{kc} becomes:
\begin{equation}
	\label{kc_alpha=1}
	k_c=\left[\Gamma\left(\varepsilon_H-\gamma(1-\beta u^*)+\frac{\sqrt{\delta}}{2}\right)\right]^{1/\alpha_1},
\end{equation}
with
\begin{equation}
	\label{delta_alpha=1}
	\delta = 4\gamma(1-\beta u^*)[\gamma(1-\beta u^*)-\varepsilon_H].
\end{equation}
Thus the characteristic spatial scale of the emerging pattern depends both on the domain parameter $\Gamma$ and on the superdiffusion exponent.

Since we assume periodic boundary conditions on $\Omega=[0,2\pi]$, pattern formation requires $k_c\geq 1$, which implies:
\begin{equation}
	\label{cond_kc>1}
	\Gamma\geq\frac{1}{\varepsilon_H-\gamma(1-\beta u^*)+\frac{\sqrt{\delta}}{2}},
\end{equation}
where the denominator is positive under the hypotheses of Proposition~\ref{propkinstabmono}, as can be readily verified. Therefore, pattern formation is possible only if $\Gamma$ exceeds the threshold given in~\eqref{cond_kc>1}.
\begin{figure*}
	\centering
	\subfigure[]{\includegraphics[width=0.4\textwidth]{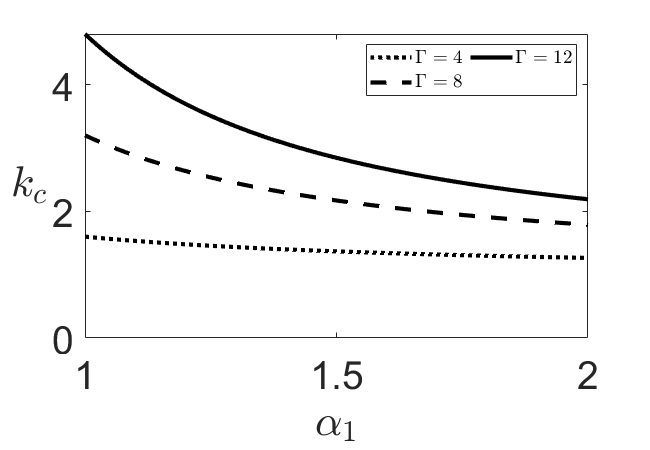}
		\label{k_vs_gamma1=gamma2}}
	\subfigure[]{\includegraphics[width=0.4\textwidth]{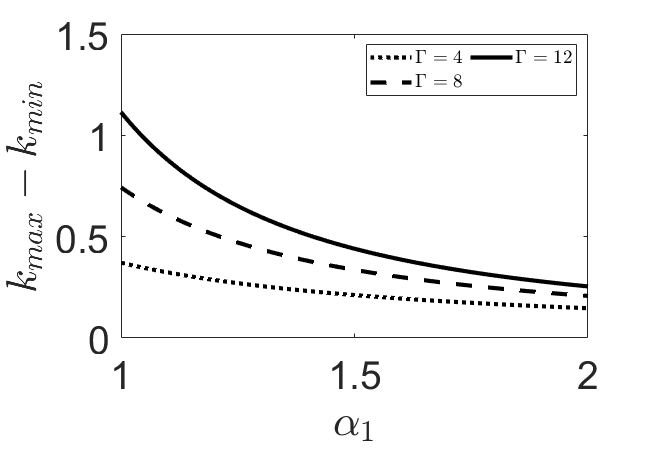}
		\label{dist_kmaxmin_vs_gamma1=gamma2}}
	\caption{(a) Critical wavenumber $k_c$ as a function of $\alpha_1=\alpha_2$, for different values of $\Gamma$. (b) Length of the instability interval $[k_{\max}-k_{\min}]$, representing the range of unstable modes, as a function of $\alpha_1=\alpha_2$ for different values of $\Gamma$. The other system parameters are fixed at $a=-0.1$, $\beta=0.3$, $\gamma=2$, $\varepsilon=1$, and $d_c=6.9525$, while $d$ is set to $d=7.0220>d_c$. }
	\label{kcr_vs_gamma1=gamma2}
\end{figure*}


\paragraph{Effect of superdiffusion on scale selection.}
Figures~\ref{kcr_vs_gamma1=gamma2} illustrates how the critical wavenumber $k_c$ and the band of unstable modes depend on $\alpha_1=\alpha_2$ and on $\Gamma$. For fixed $\Gamma$ large enough to satisfy~\eqref{cond_kc>1}, decreasing $\alpha_1=\alpha_2$ (i.e.\ strengthening superdiffusion) leads to an increase of $k_c$, and therefore to patterns with shorter characteristic wavelengths and higher spatial segregation. At the same time, the distance between the two roots of $h(k)$, whose expression is given in~\eqref{h(k)}, increases, implying a wider band of unstable modes and thus a higher potential complexity of the resulting pattern. This effect becomes more pronounced as $\Gamma$ increases above the threshold~\eqref{cond_kc>1}, see Fig.\ref{dist_kmaxmin_vs_gamma1=gamma2} and also  Fig.~\ref{dispersion_gamma1=gamma2}, where the dispersion curve and the corresponding instability region progressively shift towards larger wavenumbers, leading to an increase in both $k_c$ and the number of unstable modes. Since $\Gamma$ scales as $\Gamma\sim \ell^{\alpha_1}$ (see \eqref{parameter_adim}), this behavior reflects the fact that larger spatial domains allow for finer spatial structures, without affecting the instability threshold $d_c$.

\paragraph{Interpretation of the results.}
When both species superdiffuse with the same exponent, anomalous diffusion does not alter the classical activator--inhibitor mechanism responsible for pattern formation: instability still requires the inhibitor to diffuse faster than the activator. However, superdiffusion plays a crucial role in the selection of the characteristic spatial scale of the pattern. Stronger superdiffusion (smaller $\alpha_1=\alpha_2$) reduces the regularizing effect of the transport operator, leading to larger critical wavenumbers and broader bands of unstable modes. As a consequence, the emerging patterns are more strongly segregated and display a richer spatial structure.
Interpreting the system, for instance, as a prey--predator model, this corresponds to a situation in which both species perform Lévy-like movements with comparable statistics. Although the qualitative mechanism generating the patterns remains unchanged, increasing the strength of superdiffusion enhances spatial fragmentation and leads to more complex coexistence structures.

\subsection{Case $\alpha_1\neq \alpha_2$}

We now turn to the anomalous case in which the two species are characterized by different superdiffusive exponents, i.e.\ $\alpha_1 \neq \alpha_2$, or equivalently $\alpha \neq 1$. In this regime, transport does not simply rescale spatial patterns, but introduces an intrinsic asymmetry between the activator and the inhibitor, which may modify the classical Turing mechanism.

In this case, the Turing threshold $d_c$ given in~\eqref{dc} depends not only on the kinetic parameters, but also on the domain parameter $\Gamma$ and on the ratio $\alpha=\alpha_1/\alpha_2$. As observed in Remark~\ref{rem1}, although the two species are characterized by distinct fractional orders, the instability threshold depends only on their ratio and not on their individual values. This property allows for a systematic classification of the pattern-forming regimes in terms of the parameter $\alpha$.

\paragraph{Breaking the classical activator--inhibitor paradigm.}

For standard diffusion, i.e.\ $\alpha_1=\alpha_2=2$, the classical Turing mechanism requires the inhibitor to diffuse faster than the activator, namely $D_V>D_U$, which corresponds to $d>1$ in the dimensionless formulation~\eqref{parameter_adim}. This condition embodies the principle of short-range activation and long-range inhibition underlying pattern formation.

In the presence of anomalous diffusion, however, the effective transport of each species is determined by two independent factors: the diffusion exponent, characterized by $\alpha_1$ and $\alpha_2$, and the diffusion rate, described by the dimensionless parameter $d$, which depends on the ratio $D_V/D_U$ through the scaling~\eqref{parameter_adim}. As a consequence, the notion of ``faster'' or ``slower'' diffusion becomes more subtle.

Using the scaling~\eqref{parameter_adim}, the classical condition $D_V<D_U$ is equivalent to:
\begin{equation}
	\label{cond_DV<DU}
	d < \ell^{\alpha_1 - \alpha_2},
\end{equation}
so that pattern formation may occur even when the activator diffuses faster than the inhibitor in the classical sense, provided that the difference in the superdiffusive exponents compensates for this effect.

\begin{remark}
	Notice that the condition $d_c<\ell^{\alpha_1-\alpha_2}$ is possible if and only if 
	the following condition among parameters holds:

		\begin{equation}\label{cond_dc_smaller1}
			\Gamma^{\frac{1-\alpha}{\alpha}}>\ell^{\alpha_1 \frac{1-\alpha}{\alpha}}
			\frac{\varepsilon\left[(1-\beta u^*)\gamma(1-\alpha)+\sqrt{\delta}\right]}{\left[\varepsilon_H -\frac{\gamma}{2}(1+\alpha)(1-\beta u^*)+\frac{\sqrt{\delta}}{2}\right]^{\frac{1}{\alpha}} \left[\gamma(1+\alpha)(1-\beta u^*)-\sqrt{\delta}\right]},
		\end{equation}

	with $\delta$ given in~\eqref{delta}, and $\ell$ denoting the length of the one-dimensional spatial domain.
	
	If $\alpha>1$, the function $\Gamma^{(1-\alpha)/\alpha}$ is decreasing in $\Gamma$, and sufficiently small values of $\Gamma$ are required. Conversely, if $\alpha<1$, the same function is increasing, and sufficiently large values of $\Gamma$ are needed.
\end{remark}

\begin{figure}
	\centering
	\includegraphics[width=0.4\textwidth]{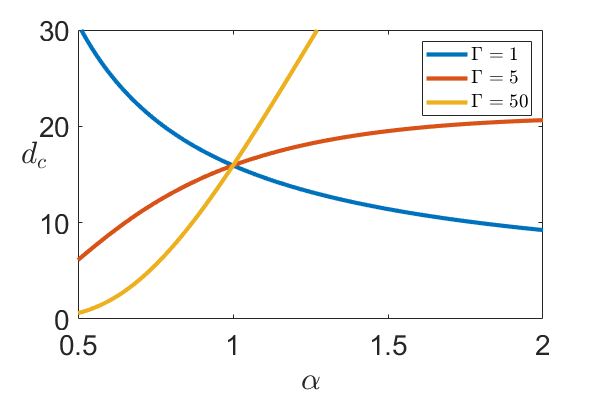}
	\caption{Dependence of the bifurcation threshold $d_c$ on $\alpha$ for different values of $\Gamma$.  The other parameters are $a=-0.1$, $\beta=0.3$, $\gamma=2.3$, and $\varepsilon=2$.}
	\label{dc_vs_alpha_multipleGamma}
\end{figure}

Figure~\ref{dc_vs_alpha_multipleGamma} illustrates the dependence of the bifurcation threshold $d_c$ on the ratio $\alpha$ for different values of $\Gamma$, with all the other kinetic parameters fixed. For sufficiently small values of $\Gamma$, $d_c(\alpha)$ is a decreasing function, whereas increasing $\Gamma$ leads to a progressively steeper increasing behavior.

This behavior can be understood from~\eqref{dc}, which shows that $d_c\propto \Gamma^{(\alpha-1)/\alpha}$. Since instability requires $d>d_c$, pattern formation is favored when $d_c$ is small. For $\alpha>1$, the exponent $(\alpha-1)/\alpha$ is positive and $d_c$ increases with $\Gamma$, making instability more difficult for large domains. Conversely, for $\alpha<1$, the exponent is negative and $d_c$ decreases with $\Gamma$, thereby facilitating pattern formation.

All curves intersect at $\alpha=1$, corresponding to the standard diffusion case $\alpha_1=\alpha_2=2$, for which $d_c$ depends only on the kinetic parameters through~\eqref{dc_alpha=1}. This point marks the transition between the classical and anomalous regimes.

These results highlight the interplay between diffusion rates, superdiffusive exponents, and reaction kinetics in determining the onset of instability.

\paragraph{Case $\alpha_1>\alpha_2$}\label{gamma1magg}

We first consider the case $\alpha_1>\alpha_2$, i.e.\ $\alpha>1$, corresponding to a regime in which the inhibitor is more strongly superdiffusive than the activator, since smaller diffusion exponents are associated with stronger anomalous effects.

In a prey--predator interpretation, this corresponds to a predator performing Lévy-like movements with a higher probability of long jumps, while the prey exhibits a more Brownian-like dynamics.

\begin{figure*}
	\subfigure[]{\includegraphics[width=0.32\textwidth]{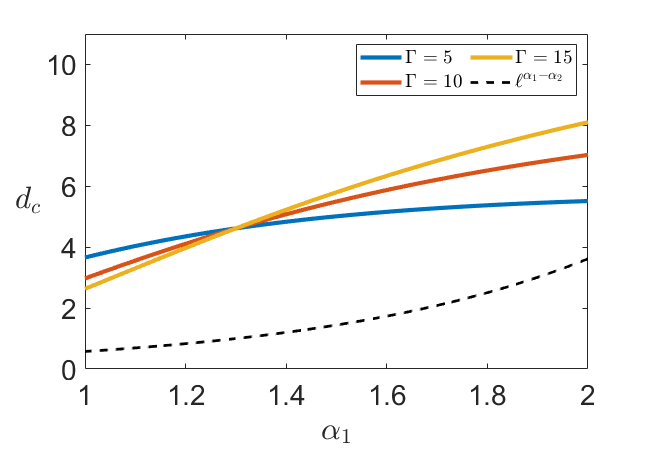}
		\label{dc_vs_alpha1_multipleGamma_alpha1bigger_alpha2=1_3}}
	\subfigure[]{\includegraphics[width=0.32\textwidth]{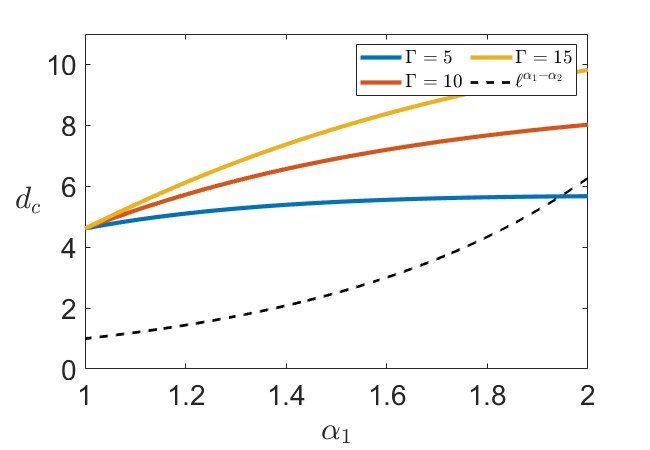}
		\label{dc_vs_alpha1_multipleGamma_alpha1bigger_alpha2=1_001}}
	\subfigure[]{\includegraphics[width=0.32\textwidth]{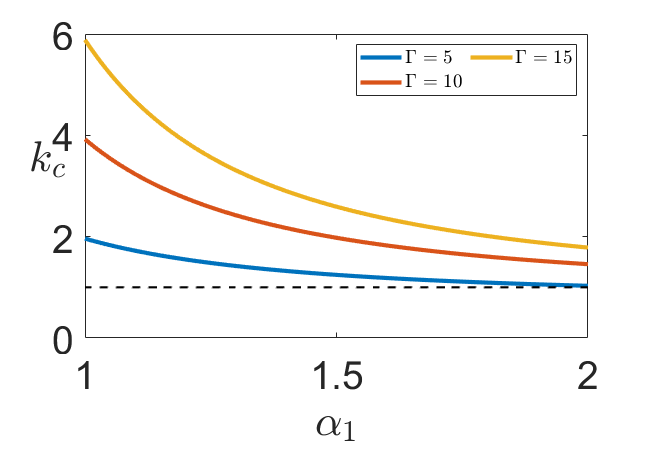}
		\label{kc_vs_alpha1_multipleGamma_alpha1bigger}}
	\caption{Dependence of the bifurcation threshold $d_c$ on $\alpha_1$ for $\alpha_2=1.3$ (a) and $\alpha_2=1.001$ (b), and of the critical wavenumber $k_c$ on $\alpha_1$ for different values of the width parameter $\Gamma$ with $\alpha_2=1.001$ (c). The dashed line represents the curve $d=\ell^{\alpha_1-\alpha_2}$, corresponding to the condition $D_V/D_U=1$ in the original variables.
		The other parameters are fixed as $a=0.1$, $\beta=0.4$, $\gamma=1.8$, and $\varepsilon=1.01$.
	}
	\label{dc_kc_vs_alpha1_multipleGamma_alpha1bigger}
\end{figure*}

Figures~\ref{dc_vs_alpha1_multipleGamma_alpha1bigger_alpha2=1_3} and~\ref{dc_vs_alpha1_multipleGamma_alpha1bigger_alpha2=1_001} show the dependence of the bifurcation threshold $d_c$ on the activator diffusion exponent $\alpha_1$, for different values of $\Gamma$ and fixed values of $\alpha_2$. In all considered cases, $d_c$ is an increasing function of $\alpha_1$, indicating that smaller values of $\alpha_1$, corresponding to stronger superdiffusion of the activator, favor pattern formation by lowering the instability threshold.

Moreover, for fixed $\alpha_1$, increasing $\Gamma$ leads to larger values of $d_c$, making instability more difficult to achieve. Conversely, smaller values of $\Gamma$ systematically reduce the threshold and promote pattern formation.

The dashed line marks the parameter region in which the classical condition $D_V>D_U$ is violated, according to~\eqref{cond_DV<DU}. In this regime, instability occurs even when the activator diffuses faster than the inhibitor. This is made possible by the scaling
$d=\frac{D_V}{D_U}\ell^{\alpha_1-\alpha_2}$, which allows the difference in diffusion exponents to compensate for the ratio of diffusion rates.

In particular, for sufficiently small values of $\alpha_2$, $d_c$ may become low enough to allow $D_U>D_V$, as shown in Fig.~\ref{dc_vs_alpha1_multipleGamma_alpha1bigger_alpha2=1_001} for $\Gamma=5$.

This behavior reflects a compensation mechanism between diffusion rate and diffusion exponent: although the activator diffuses faster in terms of rate, the stronger superdiffusion of the inhibitor balances this effect, allowing pattern formation beyond the classical regime.

The corresponding critical wavenumber must also be examined to assess whether spatial patterns can develop on bounded domains. As shown in Fig.~\ref{kc_vs_alpha1_multipleGamma_alpha1bigger}, in the case $\alpha_2=1.001$ and $\Gamma=5$, $k_c$ remains larger than unity, allowing for pattern formation on $\Omega=[0,2\pi]$ with periodic boundary conditions, provided that the instability band is sufficiently wide.

\begin{figure*}
	\subfigure[]{\includegraphics[width=0.3\textwidth]{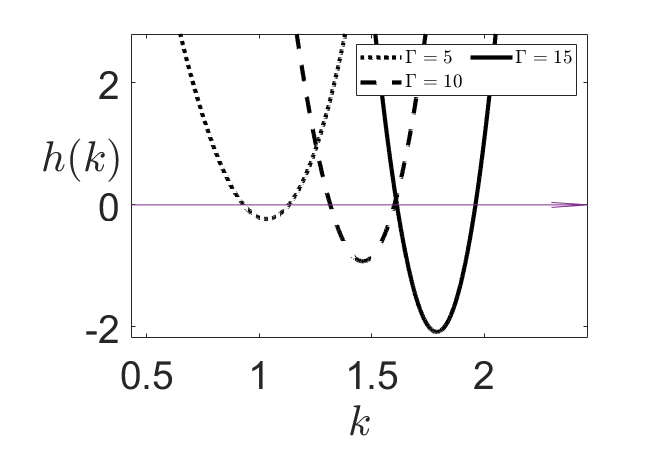}
		\label{disp_alpha1=2_alpha1magg}}
	\subfigure[]{\includegraphics[width=0.3\textwidth]{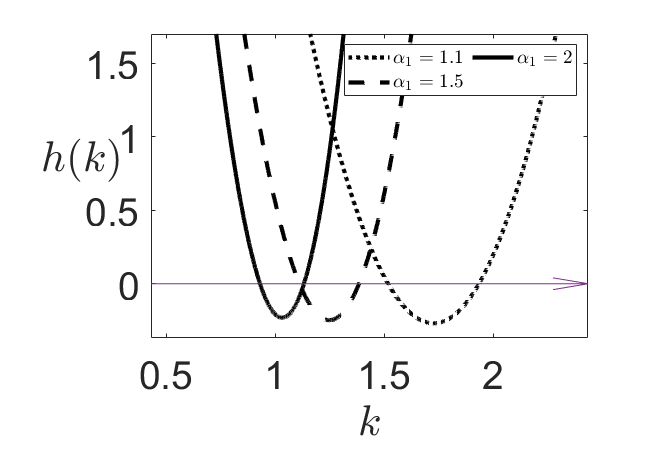}
		\label{disp_Gamma=5_alpha1magg}}
	\subfigure[]{\includegraphics[width=0.36\textwidth]{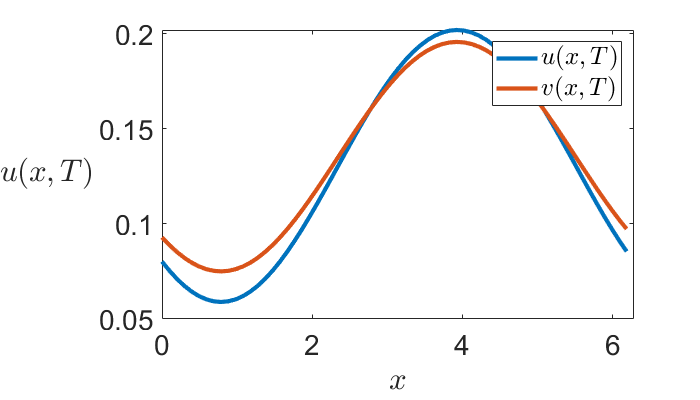}
		\label{Gamma=5_1D}}
	\caption{Dispersion curves for $\alpha_2=1.001$: (a) varying $\Gamma$ with $\alpha_1=2$, (b) varying $\alpha_1>\alpha_2$ with $\Gamma=5$, and (c) corresponding one-dimensional pattern for $\alpha_1=2$, $\alpha_2=1.001$, and $\Gamma=5$. The remaining parameters are fixed at $a=0.1$, $\beta=0.4$, $\gamma=1.8$, and $\varepsilon=1.01$, and $d=d_c(1+0.1^2)$ in all panels.
	}
	\label{disp_alpha1magg}
\end{figure*}

Figure~\ref{disp_alpha1magg} illustrates the corresponding dispersion curves. For fixed $\alpha_1=2$, increasing $\Gamma$ shifts the instability region towards larger wavenumbers and widens the band of unstable modes (Fig.~\ref{disp_alpha1=2_alpha1magg}), leading to more segregated and complex spatial patterns.

Conversely, for fixed $\Gamma=5$, increasing $\alpha_1$ shifts the dispersion curve towards smaller wavenumbers and reduces the instability band (Fig.~\ref{disp_Gamma=5_alpha1magg}), resulting in less segregated and smoother spatial structures.

These results confirm that, in this regime, both the domain parameter $\Gamma$ and the interplay between diffusion rates and superdiffusive exponents play a central role in determining the instability threshold and the characteristic spatial scale of the emerging patterns.
\paragraph{Case $\alpha_1<\alpha_2$}

We now consider the case $\alpha_1<\alpha_2$, i.e.\ $\alpha<1$, corresponding to a regime in which the activator is more strongly superdiffusive than the inhibitor.

\begin{figure*}
	\subfigure[]{\includegraphics[width=0.32\textwidth]{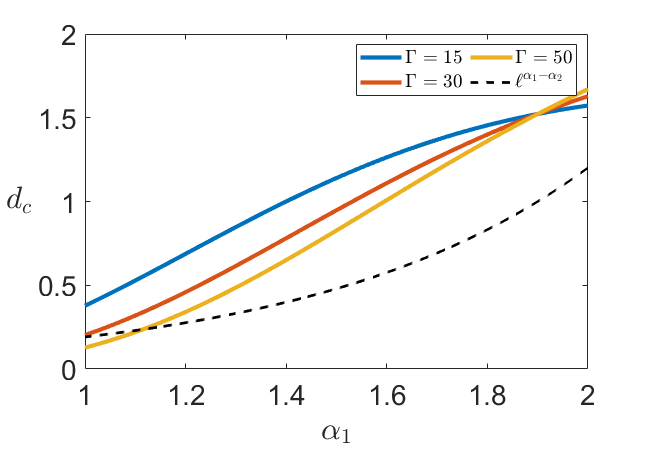}
		\label{dc_vs_alpha1_multipleGamma_alpha1smaller_alpha2=1_9}}
	\subfigure[]{\includegraphics[width=0.32\textwidth]{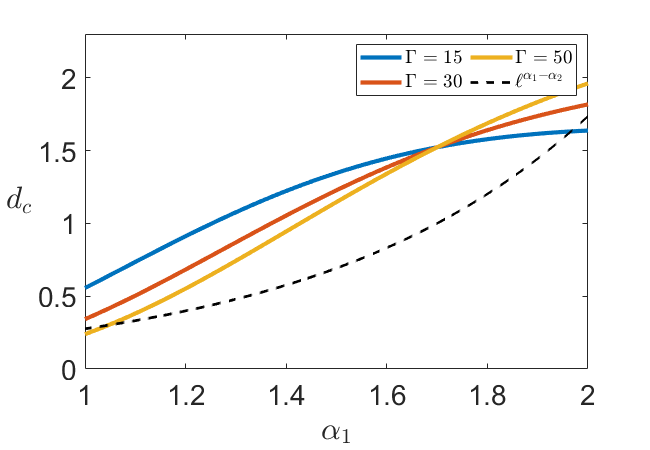}
		\label{dc_vs_alpha1_multipleGamma_alpha1smaller_alpha2=1_7}}
	\subfigure[]{\includegraphics[width=0.32\textwidth]{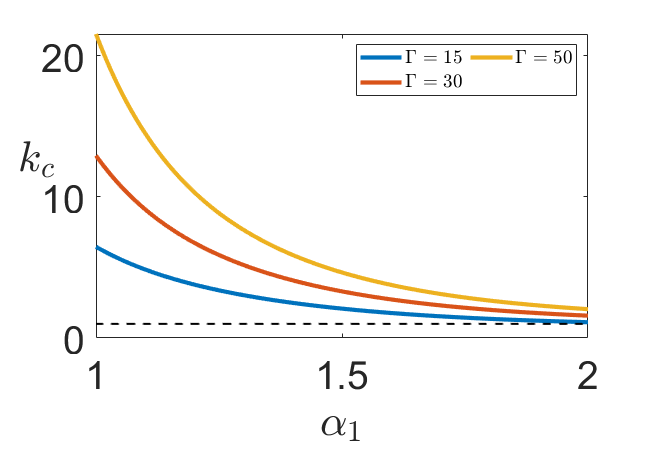}
		\label{kc_vs_alpha1_multipleGamma_alpha1smaller}}
	\caption{Dependence of the bifurcation threshold $d_c$ on $\alpha_1$ for $\alpha_2=1.9$ (a) and $\alpha_2=1.7$ (b), and of the critical wavenumber $k_c$ on $\alpha_1$ for different values of the width parameter $\Gamma$ with $\alpha_2=1.7$ (c). The dashed line represents the curve $d=\ell^{\alpha_1-\alpha_2}$, corresponding to the condition $D_V/D_U=1$. The remaining parameters are fixed at $a=0.001$, $\beta=0.1$, $\gamma=1.03$, and $\varepsilon=1.1$.
	}
	\label{dc_kc_vs_alpha1_multipleGamma_alpha1smaller}
\end{figure*}

Figures~\ref{dc_vs_alpha1_multipleGamma_alpha1smaller_alpha2=1_9} and~\ref{dc_vs_alpha1_multipleGamma_alpha1smaller_alpha2=1_7} show the dependence of the bifurcation threshold $d_c$ on $\alpha_1$ for different values of $\Gamma$ and fixed values of $\alpha_2$. For sufficiently large values of $\Gamma$, $d_c$ is an increasing function of $\alpha_1$, indicating that stronger superdiffusion of the activator (smaller $\alpha_1$) promotes instability.

Moreover, increasing $\Gamma$ systematically lowers the instability threshold when $\alpha_1\ll\alpha_2$, showing that large domains strongly favor pattern formation in this regime.

The dashed line marks the parameter region in which the classical condition $D_V>D_U$ is violated, according to~\eqref{cond_DV<DU}. In this case, pattern formation occurs even when the activator diffuses faster than the inhibitor. For instance, Fig.~\ref{dc_vs_alpha1_multipleGamma_alpha1smaller_alpha2=1_7} shows that, for $\Gamma=50$ and sufficiently small $\alpha_1$, $d_c$ becomes low enough to allow $D_U>D_V$.

In contrast with the case $\alpha_1>\alpha_2$, here the violation of the classical condition is not due to a compensation between diffusion rate and diffusion exponent, but to a combined effect: the activator diffuses faster both in terms of rate and superdiffusive exponent. This leads to a strongly non-classical instability mechanism.

\begin{figure*}
	\subfigure[]{\includegraphics[width=0.3\textwidth]{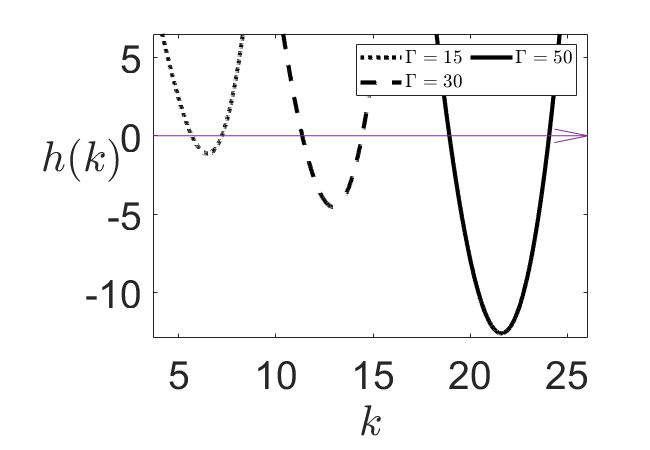}
		\label{disp_alpha2=1_7_alpha1min}}
	\subfigure[]{\includegraphics[width=0.3\textwidth]{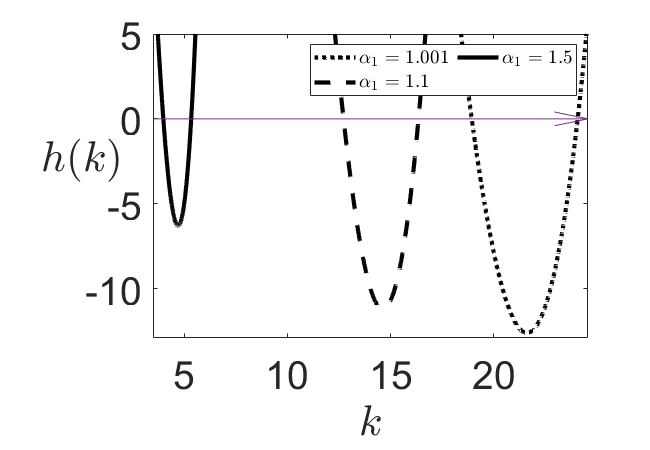}
		\label{disp_Gamma=50_alpha1min}}
	\subfigure[]{\includegraphics[width=0.36\textwidth]{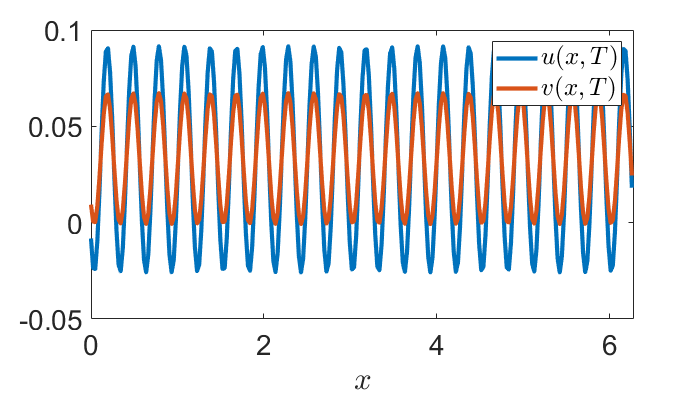}
		\label{Gamma=50_1D}}
	\caption{Dispersion curves for $\alpha_2=1.7$: (a) varying $\Gamma$ with $\alpha_1=1.001$, (b) varying $\alpha_1<\alpha_2$ with $\Gamma=50$, and (c) corresponding one-dimensional pattern for $\alpha_1=1.001$, $\alpha_2=1.7$, and $\Gamma=50$. The remaining parameters are fixed at $a=0.001$, $\beta=0.1$, $\gamma=1.03$, and $\varepsilon=1.1$, and $d=d_c(1+0.1^2)$ in all panels.
	}
	\label{disp_alpha1min}
\end{figure*}

Figure~\ref{disp_alpha1min} illustrates the corresponding dispersion curves. For fixed $\alpha_1=1.001$, increasing $\Gamma$ shifts the instability region towards larger wavenumbers and enlarges the band of unstable modes (Fig.~\ref{disp_alpha2=1_7_alpha1min}), leading to highly segregated patterns.

Conversely, for fixed $\Gamma=50$, increasing $\alpha_1$ shifts the dispersion curve towards smaller wavenumbers and reduces the instability band (Fig.~\ref{disp_Gamma=50_alpha1min}), resulting in smoother spatial structures.

As in the case $\alpha_1>\alpha_2$, the characteristic wavelength is mainly controlled by the activator diffusion exponent and by the domain parameter $\Gamma$, while the ratio $\alpha$ primarily affects the instability threshold.

In this regime, the scaling $d=\frac{D_V}{D_U}\ell^{\alpha_1-\alpha_2}$ implies that, since $\alpha_1-\alpha_2<0$, large values of $\Gamma$ and strong superdiffusion of the activator are required to compensate for the diffusion rate ratio. This explains why pattern formation with $D_V/D_U<1$ is possible only for sufficiently large domains.

From a biological viewpoint, this regime corresponds to a situation in which the prey diffuses much faster and performs frequent long jumps, while the predator moves more slowly and follows an almost Brownian dynamics. The reduced regularizing effect of strong superdiffusion leads to the formation of highly fragmented and clustered spatial structures.

\section{Weakly Nonlinear Analysis}
\label{WNL}

In this Section we perform a weakly nonlinear (WNL) analysis in a neighborhood of the Turing bifurcation threshold to describe the asymptotic behaviour of the emerging patterns. We derive the Stuart--Landau equation governing the slow temporal evolution of the dominant mode amplitude, classify the bifurcation as supercritical or subcritical, and investigate how the nonlinear regime depends on the anomalous diffusion exponents. Numerical simulations are also presented to illustrate the influence of the diffusion power on pattern morphology.


We employ a perturbative analysis based on the method of multiple scales. Let $\eta$ be a small control parameter measuring the dimensionless distance from the Turing bifurcation threshold $d_c$. Close to the Turing bifurcation, several time scales
$
T_j = \eta^j t, \ j=1,2,\dots
$
must be introduced, so that time derivatives decompose as:
\begin{equation}
	\label{time}
	\frac{\partial}{\partial t} = \eta \frac{\partial}{\partial T_1} + \eta^2 \frac{\partial}{\partial T_2} + \eta^3 \frac{\partial}{\partial T_3} + \dots
\end{equation}
Both the solution of the system \eqref{origsyst_1}--\eqref{origsyst_2} and the bifurcation parameter $d$ are expanded as:

\begin{eqnarray}
	\mathbf{w}&=&\eta\mathbf{w}_1+\eta^2\mathbf{w}_2+\eta^3\mathbf{w}_3+\mathcal{O}(\eta^4),
	\label{ex2}\\
	d&=&d_c+\eta d_1+\eta^2 d_2+\eta^3 d_3+\mathcal{O}(\eta^4).
	\label{exd}
\end{eqnarray}

%
%
Substituting \eqref{ex2}, \eqref{exd},  and \eqref{time} into \eqref{origsyst_1}--\eqref{origsyst_2} and expanding in powers of $\eta$, we obtain the following hierarchy of problems at each order in $\eta$:
\begin{subequations}
	\begin{align}
		O(\eta):& \qquad \mathscr{L}^{d_c}\mathbf{w}_1=\mathbf{0},\label{wlin1}\\
		O(\eta^2):& \qquad \mathscr{L}^{d_c}\mathbf{w}_2=\mathbf{F},\label{wlin2}\\
		O(\eta^3):& \qquad \mathscr{L}^{d_c}\mathbf{w}_3=\mathbf{G},\label{wlin3}
	\end{align}
\end{subequations}
where we have defined the linear operator $\mathscr{L}$ as follows:
\begin{equation}
	\label{linop}
	\mathscr{L}= \Gamma J+\mathscr{D},\ {\rm with\ }\mathscr{D}=\left(\begin{array}{ll}\nabla^{\alpha_1}&0\\0&d\nabla^{\alpha_2}\end{array}\right)
\end{equation}
and $\mathscr{L}^{d_c}$ denotes the operator $\mathscr{L}$ evaluated at $d=d_c$. The source terms $\mathbf{F}$ and $\mathbf{G}$ are:

	\begin{eqnarray}
		\mathbf{F}&=  &\frac{\partial\mathbf{w}_1 }{\partial T_1} - 
		\left(\begin{array}{c}
			0\\ d_1 \nabla^{\alpha_2} v_1 
		\end{array}\right)+\Gamma\left(\begin{array}{c}
			3u^*u_1^2-\beta u_1v_1\\ 0 
		\end{array}\right),\\\label{G}
		\mathbf{G}&=&  \frac{\partial\mathbf{w}_1 }{\partial T_2} + \frac{\partial\mathbf{w}_2 }{\partial T_1}- 
		\left(\begin{array}{c}
			0\\ d_2 \nabla^{\alpha_2} v_1
		\end{array}\right) - \left(\begin{array}{c}
			0\\ d_1 \nabla^{\alpha_2} v_2
		\end{array}\right) + \Gamma \left(\begin{array}{c}
			6u^*u_1u_2-\beta (u_1v_2+u_2v_1)+u_1^3\\ 0 
		\end{array}\right),
	\end{eqnarray}

where $\mathbf{w}_i\equiv(u_i, v_i)^\top,\, i=1,2$.
The solution to the linear problem \eqref{wlin1} satisfying the periodic
boundary conditions is given by:
\begin{equation}
	\mathbf{w}_1=A(T_i)\bm{\varrho}\hspace{.1cm} e^{ik_c x} + c.c., 
	\label{w1}
\end{equation}
%
%
%
where $A(T_i)$ is the still unknown amplitude, $D_{k_c}^{d_c} $ is the matrix in \eqref{vectsyst} evaluated at $k\equiv k_c$ and $d\equiv d_c$,  and 
$\bm{\varrho}\in {\rm Ker}\left\{\Gamma J-D_{k_c}^{d_c}\right\}$ reads:%
\begin{equation}\label{rho}
	\bfrho=z \left(\begin{matrix}
		1\\
		M \end{matrix} \right), \quad z\in \mathbb{C}, \quad M = \frac{\Gamma\gamma\varepsilon}{\varepsilon\Gamma+d_ck_c^{\alpha_2}}.
\end{equation}		
Once the solution \eqref{w1} is substituted into \eqref{wlin2} at $O(\eta^2)$, the solvability condition is satisfied by imposing $T_1=0$ and $d_1=0$. Moreover, since the fractional Laplacian operators $\nabla^{\alpha_i}$ are self-adjoint in $L^2$ \cite{Muratori2016,Petrosyan2020}, the adjoint problem can be treated in the same way as in the classical case. Therefore, the problem \eqref{wlin2} admits a solution $\ww_2$ of the following form:
\begin{equation}
	\mathbf{w}_2 = |A|^2\mathbf{w}_{20}+A^2\mathbf{w}_{22}e^{2ik_cx}+c.c.,
	\label{w2}
\end{equation}
where $\mathbf{w}_{20}$ and $\mathbf{w}_{22}$ solve the following linear systems:
\[
\begin{cases}
	J\mathbf{w}_{20} = \Gamma  \left( 3u^*- \dfrac{ \beta \gamma \varepsilon\Gamma}{\varepsilon\Gamma+d_ck_c^{\alpha_2}}, 0 \right)^{\top}, \\
	\left(J - D_{2k_c} \right)  \mathbf{w}_{22} = \Gamma  \left( 3u^*- \dfrac{ \beta \gamma \varepsilon\Gamma}{\varepsilon\Gamma+d_ck_c^{\alpha_2}}, 0 \right)^{\top}.
\end{cases}
\label{w2_systems}
\]
Substituting \eqref{w1} and \eqref{w2} into the expression of $\mathbf{G}$ in \eqref{G} yields:

	\begin{gather}
		\mathbf{G}=  \left(\frac{\partial A }{\partial T_2} \bm{\varrho} + A \mathbf{G}_1 + A|A|^2 \mathbf{G}_3\right) e^{ik_c x}  + \mathbf{G}^* + c.c.,  \\
		{\rm where}\quad	\mathbf{G}_1 = \left(\begin{array}{c}
			0\\d_2k_c^{\alpha_2} M
		\end{array}\right),  
		\quad	\mathbf{G}_3 = \Gamma \left(\begin{array}{c}
			(2 u_{20}+u_{22})(6u^*-\beta M)-\beta(2v_{20}+v_{22})+3\\0
		\end{array}\right),\nonumber
	\end{gather}

where $\mathbf{w}_{2j}\equiv (u_{2j}, v_{2j})^\top$.
Imposing the Fredholm solvability condition at $O(\eta^3)$, we obtain the following Stuart--Landau equation for the amplitude:
\begin{equation}
	\frac{\partial A}{\partial T_2}=\sigma A- LA|A|^2,
	\label{SL}
\end{equation}
with:
\begin{equation}
	\sigma=-\frac{\langle \mathbf{G}_1,\bm{\varrho'} \rangle}{\langle\bm{\rho},\bm{\rho'} \rangle}, \qquad L=\frac{\langle \mathbf{G}_3,\bm{\varrho'} \rangle}{\langle\bm{\varrho},\bm{\varrho'} \rangle}.
	\label{sigmaL}
\end{equation}
and $\bm{\varrho}'\in \mathrm{Ker}(\Gamma J^{\top}-D_{k_c}^{d_c})$ is given by:
\begin{equation}
	\bm{\varrho}'=h\left(\begin{matrix}
		1\\
		M' \end{matrix} \right), \quad h\in \mathbb{C}, \quad M'=\frac{\Gamma(\beta u^*-1)}{\Gamma\varepsilon+d_c k_c^{\alpha_2}}.
	\label{rho1}
\end{equation}
The coefficient $\sigma$ is always positive in the parameter region where Turing instability occurs. The sign of the Landau coefficient $L$ determines the nature of the bifurcation: for $L>0$ the bifurcation is supercritical, whereas for $L<0$ it is subcritical.
In the supercritical case, the amplitude equation predicts a stable saturation of the instability and the asymptotic amplitude of the emerging pattern can be explicitly determined. This provides a quantitative description of the weakly nonlinear regime close to threshold, as stated in Theorem~\ref{teo_super}.
In the subcritical case, the instability leads to the emergence of large-amplitude patterns and the cubic truncation of the amplitude equation is no longer sufficient to capture their dynamics. Higher-order nonlinear terms must therefore be taken into account. In particular, a consistent description of the pattern amplitude would require extending the weakly nonlinear analysis up to order $\eta^5$.
\begin{theorem}{(Asymptotic Solution in the supercritical regime)}\label{teo_super}
	Under the hypotheses of Theorem~\ref{propturininstabmono}, assume that the distance from the bifurcation threshold is sufficiently small so that the only unstable mode is the critical one with wavenumber $k_c$, given in \eqref{kc}, and that $k_c$ is compatible with the imposed boundary conditions. Moreover, assume that the initial conditions of system \eqref{origsyst_1}-\eqref{origsyst_2} are a sufficiently small perturbation of the equilibrium $E^*$. Then, the emerging pattern solution takes the form:
	\begin{equation}\label{sol_appr}
		\ww=\eps\bfrho |A_{\infty}|\cos\left(k_cx\right),
	\end{equation}
	where:
	\[
	|A_{\infty}|=\sqrt{\frac{\sigma}{L}}
	\]
	is the stable equilibrium amplitude of the Stuart--Landau equation~\eqref{SL}.
\end{theorem}
The above result can be rigorously justified by means of center manifold reduction, following the approach developed in~\cite{FGGLS2025}.

We conclude this Section by presenting numerical simulations illustrating how the diffusion power influences the nature of the Turing bifurcation.
\begin{figure}
	\centering
	{\includegraphics[width=0.45\textwidth]{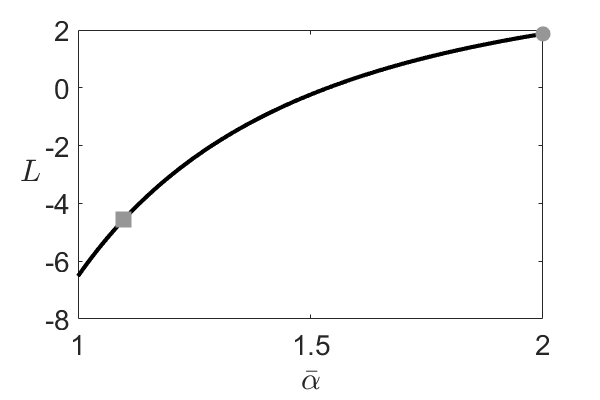}
	}
	\caption{Landau coefficient $L$ versus the diffusion exponent $\bar{\alpha}:=\alpha_1=\alpha_2$, for $a=-0.28$, $\beta=0.3$, $\gamma=2$, $\varepsilon=0.9$, and $\Gamma=13$. The sign change of $L$ marks the transition from strong superdiffusion to classical diffusion. The light gray dot and square markers correspond to the supercritical and subcritical patterns shown in Fig.~\ref{super_pattern_alpha1=alpha2=2_epssmall} and Fig.~\ref{sub_pattern_alpha1=alpha2=1_1_epssmall}, respectively.\label{Lvsalphabar}}
\end{figure}		
In Fig.~\ref{Lvsalphabar} we report the dependence of the Landau coefficient $L$ on the common diffusion exponent $\bar{\alpha}=\alpha_1=\alpha_2$. As $\bar{\alpha}$ decreases, corresponding to stronger superdiffusion, $L$ changes sign from positive to negative, indicating a transition from supercritical to subcritical bifurcation. This shows that anomalous diffusion systematically promotes subcriticality.

This effect is further illustrated in Fig.~\ref{super-sub_regions_alpha1=alpha2}, where we display the Turing instability regions in the $(a,\gamma)$-plane for fixed kinetic parameters and two different values of $\bar{\alpha}$. In the classical diffusion case ($\bar{\alpha}=2$), both supercritical and subcritical regimes are observed. When superdiffusion is enhanced ($\bar{\alpha}=1.1$), the subcritical region becomes larger, indicating a systematic, although not dominant, tendency towards subcriticality.
The corresponding patterns, shown in panels (b,d), confirm the transition from supercritical to subcritical behavior induced by stronger superdiffusion.

In all these simulations, the kinetic parameters and the distance from the instability threshold are kept fixed. In particular, the bifurcation parameter $d$ is chosen at the same relative distance from $d_c$, which is identical in both cases since it depends only on the ratio $\alpha=\alpha_1/\alpha_2=1$. Therefore, the observed changes are solely due to the variation of the diffusion exponent.

A similar tendency is observed also in asymmetric transport configurations, in which only one of the two species exhibits superdiffusive behavior while the other undergoes classical diffusion. In all the tested cases, the presence of anomalous transport in at least one component leads to an enlargement of the subcritical region, independently of whether it is associated with the activator or the inhibitor. This indicates that the promotion of subcriticality is a robust effect of superdiffusion, rather than a specific feature of the symmetric case $\alpha_1=\alpha_2$.

We have verified that the same qualitative behavior persists for other parameter sets, including regimes where Hopf bifurcation is absent and configurations in which only one species exhibits anomalous diffusion. In all cases, stronger superdiffusion favors subcriticality and enhances nonlinear effects.
\begin{figure*}\label{super-sub_regions_alpha1=alpha2}
	\subfigure[]{\includegraphics[width=0.4\textwidth]{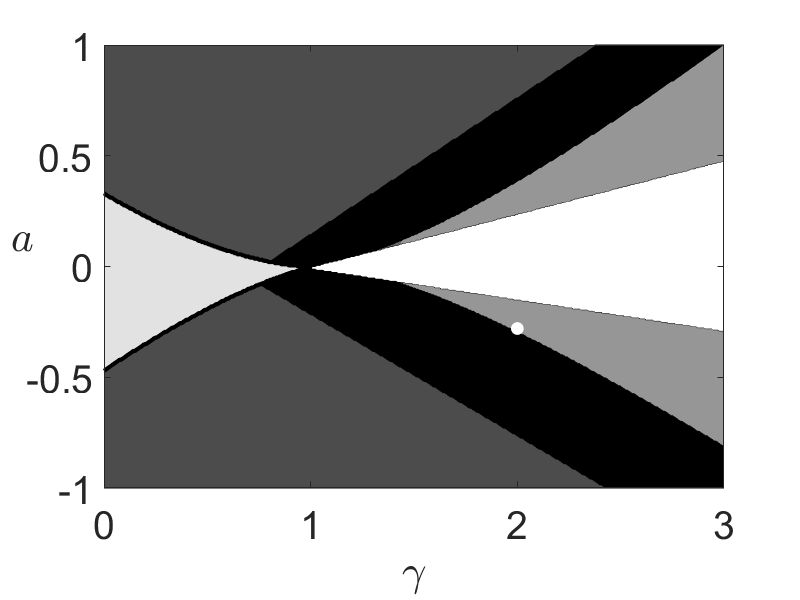}
		\label{super-sub_regions_alpha1=alpha2=2_epssmall}}
	\subfigure[]{\includegraphics[width=0.47\textwidth]{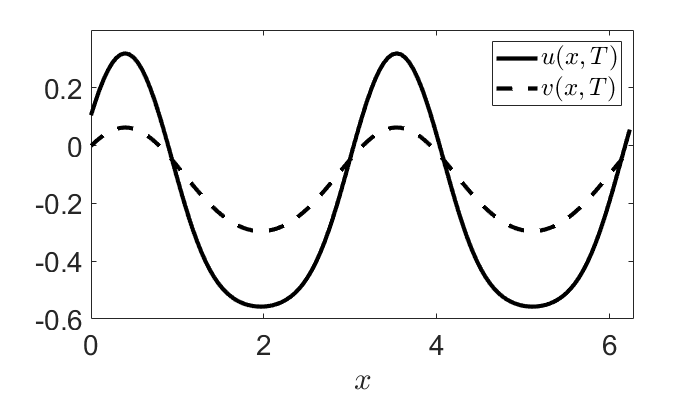}
		\label{super_pattern_alpha1=alpha2=2_epssmall}}
	\subfigure[]{\includegraphics[width=0.4\textwidth]{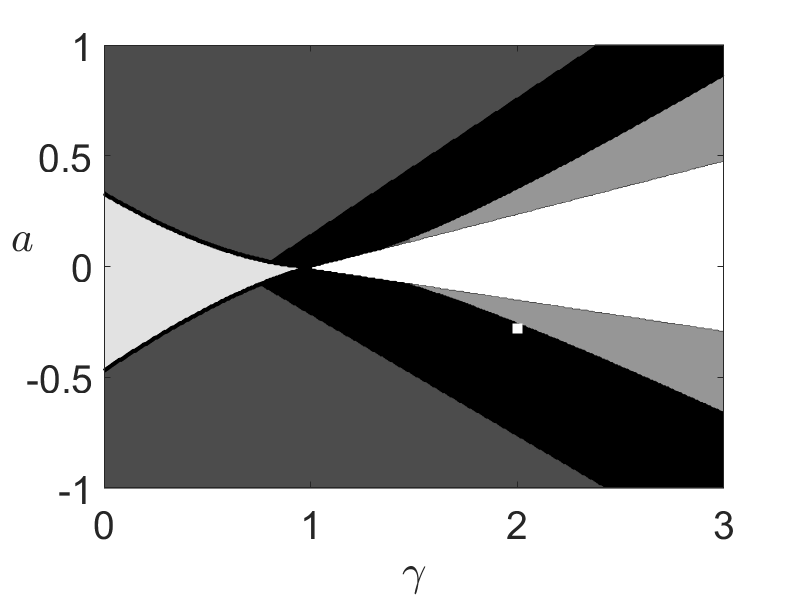}
		\label{super-sub_regions_alpha1=alpha2=1_1_epssmall}}
	\subfigure[]{\includegraphics[width=0.47\textwidth]{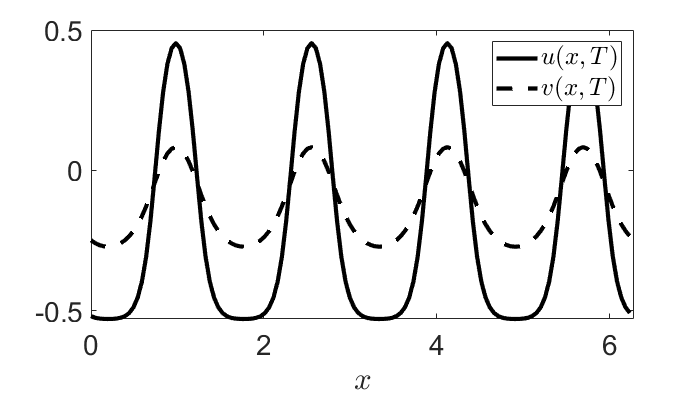}
		\label{sub_pattern_alpha1=alpha2=1_1_epssmall}}
	\caption{(a,c) Turing instability regions in the $(a,\gamma)$ parameter plane, with $\beta$, $\varepsilon$, and $\Gamma$ chosen as in Fig.~\ref{Lvsalphabar}. Panels (a,b) correspond to the classical diffusion case $\bar{\alpha}=2$, while panels (c,d) correspond to the superdiffusive case $\bar{\alpha}=1.1$. In (a,c), the light gray, dark gray, black, and medium gray areas represent the bistable, excitable, monostable subcritical, and monostable supercritical regimes, respectively.
		The dot and square markers in (a,c) correspond to the same markers in Fig.~\ref{Lvsalphabar} and indicate the parameter values for which the patterns in (b,d) are computed.
		(b,d) Turing patterns for $u$ and $v$ corresponding to the marked parameter values, with $d=d_c(1+\eta^2)$, $d_c=11.2367$, and $\eta=0.1$. The classical diffusion case (b) exhibits a supercritical bifurcation, whereas the superdiffusive case (d) leads to a subcritical bifurcation.
		\label{super-sub_patterns_epssmall}}
\end{figure*}

\section{Turing-Hopf bifurcation and its normal form}
\label{TH}
We have shown that the equilibrium $E^*$ may lose stability either through a Turing bifurcation (Theorem~\ref{propturininstabmono}) or through a Hopf bifurcation (Proposition~\ref{prophopf}), depending on the parameter values and on the anomalous diffusion exponents. As a consequence, the parameter space is divided into regions characterized by stationary patterns, temporal oscillations, or their coexistence.
In particular, when the thresholds for Turing and Hopf instabilities coincide, a codimension--2 Turing--Hopf bifurcation occurs. In a neighborhood of such a point, the emerging dynamics is influenced by both spatial and temporal instabilities. In the present fractional setting, this interaction is affected by the superdiffusive transport through the coefficients of the reduced amplitude equations.

In this Section, we study the interaction between Turing and Hopf modes near a codimension--2 bifurcation point. By means of a multiple-scales perturbation approach, we derive the associated normal form and determine the parameter regions in which different dynamical regimes arise. We also present numerical simulations illustrating the main theoretical predictions.

In analogy with the weakly nonlinear analysis developed in Section~\ref{WNL}, we investigate the dynamics in a neighborhood of the codimension--2 Turing--Hopf bifurcation point by means of a multiple-scales expansion. 

As in the purely Turing case, the slow temporal scales
$T_j$ are introduced in order to describe the slow modulation of the pattern amplitude. However, in contrast to the Turing bifurcation, near a Turing--Hopf point the equilibrium $E^*$ undergoes temporal oscillations associated with the Hopf instability. As a consequence, the fast time scale $T_0$
must also be retained in the expansion and
the time derivatives decompose as:
\begin{equation}
	\label{time0}
	\frac{\partial}{\partial t} = \frac{\partial}{\partial T_0} + \eta \frac{\partial}{\partial T_1} + \eta^2 \frac{\partial}{\partial T_2} + \eta^3 \frac{\partial}{\partial T_3} + \dots
\end{equation}

Without loss of generality, we choose the parameter $\varepsilon$ to measure the distance from the Hopf bifurcation and the parameter $d$ to measure the distance from the Turing bifurcation. Accordingly, the solution and the parameter $d$ are expanded as in \eqref{ex2} and \eqref{exd}, while $\varepsilon$ is expanded as follows:
\begin{equation}
	\varepsilon = \varepsilon_c + \eta \varepsilon^{(1)} + \eta^2 \varepsilon^{(2)} + \eta^3 \varepsilon^{(3)} + O(\eta^4).
	\label{expa}
\end{equation}
All the coefficients $\mathbf{w}_i$ now depend on the multiple time scales $T_j$, $j=0,1,2,\dots$.

Substituting \eqref{time}, \eqref{ex2},\eqref{exd} and \eqref{expa} into the system \eqref{origsyst_1}--\eqref{origsyst_2} and collecting terms at successive orders in $\eta$, we obtain a hierarchy of linear problems. At leading order, $\mathbf{w}_1$ satisfies:
\begin{equation}
	\label{linw1}
	\left(\frac{\partial}{\partial T_0}-\mathscr{L}^c\right) \mathbf{w}_1 = 0,
\end{equation}
where $\mathscr{L}^c$ is the operator defined in \eqref{linop}  evaluated at the codimension--two point $(\varepsilon_c,d_c)$.

The general solution of \eqref{linw1} is:
\begin{equation}
	\label{w1sol}
	\mathbf{w}_1 = A(T_k) e^{i k_c x} \mathbf{e}_1 + B(T_k) e^{i \Omega_c T_0} \mathbf{e}_2 + c.c.,
\end{equation}
where $A$ and $B$ are slowly varying complex amplitudes, and $c.c.$ denotes the complex conjugate. The critical wavenumber $k_c$ is given in \eqref{kc}, while:
\begin{equation}
	\label{om_val}
	\Omega_c = \gamma \sqrt{\det(J_c)}
\end{equation}
is the Hopf frequency at the codimension--two point.

The eigenvectors $\mathbf{e}_1 \in {\rm Ker}\left\{\Gamma J^{\varepsilon_c}-D_{k_c}^{d_c}\right\}$ and $\mathbf{e}_2 \in {\rm Ker}\left\{J^{\varepsilon_c} - i \Omega_c I\right\}$ are chosen as:
\begin{equation}
	\label{e1e2}
	\mathbf{e}_1 = \begin{pmatrix}
		\dfrac{J^{\varepsilon_c}_{22} - k_c^2 D^{d_c}_{22}}{J^{\varepsilon_c}_{21} - k_c^2 D^{d_c}_{21}} \\
		-1
	\end{pmatrix},
	\quad
	\mathbf{e}_2 = \begin{pmatrix}
		\dfrac{i \Omega_c - J^{\varepsilon_c}_{22}}{J^{\varepsilon_c}_{21}} \\
		1
	\end{pmatrix}.
\end{equation}
At  $O(\eta^2)$ we obtain the following linear equation for $\mathbf{w}_2$:
\begin{equation}\label{linw2}
	\left(\frac{\partial}{\partial T_0} - \mathscr{L}^c\right)\mathbf{w}_2
	= -\left(\frac{\partial}{\partial T_1} - \mathscr{L}^1\right)\mathbf{w}_1 + \mathbf{H},
\end{equation}
where:
\begin{eqnarray*}\label{L1F}
	\mathscr{L}^1 &=&
	\begin{pmatrix}
		0 & 0\\
		\Gamma \varepsilon^{(1)} \gamma & -\Gamma \varepsilon^{(1)} + d^{(1)} \nabla^{\alpha_2}
	\end{pmatrix}, \\
	\mathbf{H} &=&
	\Gamma
	\begin{pmatrix}
		u_1(\beta v_1- 3u^* u_1) \\
		0
	\end{pmatrix}.
\end{eqnarray*}
Equation \eqref{linw2} admits a solution if and only if the Fredholm alternative is satisfied. 
In order to eliminate secular terms arising from resonant components of the forcing, we impose 
$T_1=0$ and set $\varepsilon^{(1)}=d^{(1)}=0$, so that the solvability condition is automatically fulfilled.

Under these assumptions, the solution of \eqref{linw2} can be written in the form
\[
\begin{split}
	&\mathbf{w}_2 =\, |A|^2(\mathbf{w}_{200}^{(2)}+\mathbf{w}_{202} e^{2ik_c x})
	+|B|^2(\mathbf{w}_{200}^{(1)}+\mathbf{w}_{220} e^{2i\Omega_c T_0})\\
	&+
	AB\, \mathbf{w}_{211} e^{i(\Omega_c T_0 + k_c x)}
	+
	A\bar{B}\, \mathbf{w}_{2-11} e^{i(-\Omega_c T_0 + k_c x)}
	+ c.c.,
\end{split}
\]
where the coefficients $\mathbf{w}_{200}^{(i)}$, $i=1,2$, and $\mathbf{w}_{2j1}$, $j=\pm1$, 
are solutions of explicit linear systems obtained by collecting terms with the same spatio--temporal dependence in \eqref{linw2}. 
These systems are analogous to those arising in the weakly nonlinear analysis of Section~\ref{WNL} and can be computed in closed form. 
For brevity, their explicit expressions are not reported here.

At $O(\eta^3)$ we gets the following equation for $\mathbf{w}_3$:
\begin{equation}\label{linw3}
	\left(\frac{\partial}{\partial T_0} - \mathscr{L}^c\right)\mathbf{w}_3
	= -\left(\frac{\partial}{\partial T_2} - \mathscr{L}^2\right)\mathbf{w}_1 + \mathbf{P},
\end{equation}
where:
\begin{eqnarray*}\label{L2G}
	\mathscr{L}^2 &=&
	\begin{pmatrix}
		0 & 0\\
		\Gamma \varepsilon^{(2)} \gamma & -\Gamma \varepsilon^{(2)} + d^{(2)} \nabla^{\alpha_2}
	\end{pmatrix}, \\
	\mathbf{H} &=&
	\Gamma
	\begin{pmatrix}
		\beta u_2 v_1 - u_1(u_1^2+ \beta v_2- 6u^* u_2)  \\
		0
	\end{pmatrix}.
\end{eqnarray*}
Also in this case, secular terms arise in \eqref{linw3}. Imposing the solvability condition yields the following system of amplitude equations for $A$ and $B$:
\begin{eqnarray}\label{comp_eq1}
	\frac{\partial B}{\partial T_2} &=& \tilde{\sigma}_1 B - \tilde{L}_1 |B|^2 B + \tilde{\Omega}_1 |A|^2 B, \\
	\label{comp_eq2}
	\frac{\partial A}{\partial T_2} &=& \tilde{\sigma}_2 A - \tilde{L}_2 |A|^2 A + \tilde{\Omega}_2 |B|^2 A.
\end{eqnarray}
The explicit derivation of the coefficients $\tilde{\sigma}_i$, $\tilde{L}_i$, and $\tilde{\Omega}_i$, $i=1,2$ involves lengthy but straightforward algebraic computations, obtained by projecting the nonlinear terms onto the critical eigenspaces through the corresponding adjoint modes. These calculations follow the standard multiple--scale procedure diven in Section \ref{WNL} and they are omitted here.
The coefficients in \eqref{comp_eq1} are complex, whereas those in \eqref{comp_eq2} are real. Moreover, $\tilde{\sigma}_1$ depends linearly on $\varepsilon^{(2)}$, while $\tilde{\sigma}_2$ depends linearly on both $\varepsilon^{(2)}$ and $d^{(2)}$. All other coefficients are independent of $\varepsilon^{(2)}$ and $d^{(2)}$.

Writing the complex amplitudes $A$ and $B$ in polar form and separating real and imaginary parts, we finally obtain the normal form:
\begin{eqnarray}\label{normal1}
	\frac{\partial \rho_1}{\partial T_2} &=& \rho_1 \left( \sigma_1 - L_1 \rho_1^2 + \Omega_1 \rho_2^2 \right), \\
	\label{normal2}
	\frac{\partial \rho_2}{\partial T_2} &=& \rho_2 \left( \sigma_2 - L_2 \rho_2^2 + \Omega_2 \rho_1^2 \right), \\
	\label{normal3}
	\frac{\partial \vartheta_1}{\partial T_2} &=& \sigma'_1 - L'_1 \rho_1^2 + \Omega'_1 \rho_2^2, \\
	\label{normal4}
	\frac{\partial \vartheta_2}{\partial T_2} &=& 0,
\end{eqnarray}
where $\sigma_j = \Re(\tilde{\sigma}_j)$, $L_j = \Re(\tilde{L}_j)$, $\Omega_j = \Re(\tilde{\Omega}_j)$, $j=1,2$, while $\sigma'_1 = \Im(\tilde{\sigma}_1)$, $L'_1 = \Im(\tilde{L}_1)$, and $\Omega'_1 = \Im(\tilde{\Omega}_1)$.

The normal form system \eqref{normal1}--\eqref{normal4} admits the following four equilibrium solutions:
\begin{eqnarray*}
	O &:& \rho_1 = \rho_2 = 0, \label{ies} \\
	H &:& \rho_1^2 = \frac{\sigma_1}{L_1}, \quad \rho_2 = 0, \label{hb1} \\
	T &:& \rho_1 = 0, \quad \rho_2^2 = \frac{\sigma_2}{L_2}, \label{hb2} \\
	TH &:& \rho_1^2 = \frac{\sigma_1 L_2 + \sigma_2 \Omega_1}{L_1 L_2 - \Omega_1 \Omega_2}, 
	\qquad
	\rho_2^2 = \frac{\sigma_2 L_1 + \sigma_1 \Omega_2}{L_1 L_2 - \Omega_1 \Omega_2}.
	\label{qpsx}
\end{eqnarray*}
The trivial equilibrium $O$ corresponds to the spatially homogeneous steady state of the original system \eqref{origsyst_1}--\eqref{origsyst_2}.  
The solution $H$ represents a spatially homogeneous but time-periodic oscillation generated by a Hopf bifurcation.  
The solution $T$ corresponds to a stationary spatially periodic Turing pattern with critical wavenumber $k_c$.  
Finally, the solution $TH$ represents a mixed spatio-temporal state in which spatial modulation and temporal oscillations coexist.

In the following, we analyze the linear stability of the equilibria \eqref{ies}--\eqref{qpsx} and construct the corresponding bifurcation diagram in the plane $(d^{(2)}, \varepsilon^{(2)})$.

The Jacobian matrix associated with \eqref{normal1}--\eqref{normal4} evaluated at the trivial equilibrium $O$ has eigenvalues $\sigma_1$ and $\sigma_2$. Therefore, the equilibrium $O$ is stable if and only if:
\begin{equation}\label{0stable}
	\sigma_j < 0, \qquad j = 1,2.
\end{equation}

Let us now consider the equilibrium $H$. The eigenvalues of the Jacobian matrix evaluated at $H$ are $\lambda_1^H = -2\sigma_1$ and $\lambda_2^H = \sigma_2 + \sigma_1 \Omega_2 / L_1$. Therefore, the conditions for the existence and stability of $H$ are:
\begin{eqnarray}
	L_1 > 0, \label{H1stable1} \\
	\sigma_1 > 0, \label{H1stable2} \\
	\sigma_2 + \frac{\sigma_1 \Omega_2}{L_1} < 0. \label{H1stable3}
\end{eqnarray}
Along the line $\mathcal{S}_1$ defined by:
\begin{equation}\label{lineaS1}
	\sigma_1 = 0, \qquad \sigma_2 < 0,
\end{equation}
the trivial equilibrium $O$ loses stability and bifurcates into the solution $H$. This corresponds to a Hopf bifurcation of the original reaction--diffusion system \eqref{origsyst_1}--\eqref{origsyst_2}. Since $\sigma_1$ depends linearly only on $\varepsilon^{(2)}$, the line $\mathcal{S}_1$ is horizontal in the plane $(d^{(2)}, \varepsilon^{(2)})$.

We now turn to the equilibrium $T$. Its existence and stability require:
\begin{eqnarray}
	L_2 > 0, \label{H2stable1} \\
	\sigma_2 > 0, \label{H2stable2} \\
	\sigma_1 + \frac{\sigma_2 \Omega_1}{L_2} < 0. \label{H2stable3}
\end{eqnarray}
Along the critical line $\mathcal{S}_2$ defined by:
\begin{equation}\label{lineaS2}
	\sigma_2 = 0, \qquad \sigma_1 < 0,
\end{equation}
the trivial equilibrium $O$ loses stability and bifurcates into the Turing solution $T$. This corresponds to a stationary pattern-forming instability in the original system \eqref{origsyst_1}--\eqref{origsyst_2}. Since $\sigma_2$ depends linearly on both $d^{(2)}$ and $\varepsilon^{(2)}$, the line $\mathcal{S}_2$ is an oblique line in the plane $(d^{(2)}, \varepsilon^{(2)})$.

Finally, we consider the mixed-mode equilibrium $TH$. The characteristic polynomial of the Jacobian matrix evaluated at $TH$ reads $\lambda^2 - t \lambda + d$, where:
\begin{equation}\label{td}
	t = -2(L_1 \rho_1^2 + L_2 \rho_2^2), \qquad
	d = 4 \rho_1^2 \rho_2^2 (L_1 L_2 - \Omega_1 \Omega_2).
\end{equation}
The equilibrium $TH$ is stable if and only if $t < 0$ and $d > 0$. Using the explicit expressions \eqref{qpsx}, one finds that $TH$ exists and is stable provided that:
\begin{eqnarray}
	\sigma_1 L_2 + \sigma_2 \Omega_1 > 0, \label{Tstable1} \\
	\sigma_2 L_1 + \sigma_1 \Omega_2 > 0, \label{Tstable2} \\
	L_1 L_2 - \Omega_1 \Omega_2 > 0, \label{Tstable3} \\
	\sigma_1 L_2 (\Omega_2 + L_1) + \sigma_2 L_1 (\Omega_1 + L_2) > 0. \label{Tstable4}
\end{eqnarray}

From conditions \eqref{H1stable1}--\eqref{H1stable3} and \eqref{Tstable2}, we define the line $\mathcal{S}_3$:
\begin{equation}\label{lineaS3}
	\sigma_2 L_1 + \sigma_1 \Omega_2 = 0, \qquad L_1 > 0, \quad \sigma_1 > 0.
\end{equation}
Crossing $\mathcal{S}_3$, the Hopf solution $H$ loses stability and bifurcates into the mixed-mode solution $TH$. This corresponds to a Turing--Hopf bifurcation in the original system.

Similarly, from conditions \eqref{H2stable1}--\eqref{H2stable3} and \eqref{Tstable1}, we define the line $\mathcal{S}_4$:
\begin{equation}\label{lineaS4}
	\sigma_1 L_2 + \sigma_2 \Omega_1 = 0, \qquad L_2 > 0, \quad \sigma_2 > 0.
\end{equation}
Crossing $\mathcal{S}_4$, the Turing solution $T$ loses stability and bifurcates into the mixed-mode solution $TH$.

Notice that both $\mathcal{S}_3$ and $\mathcal{S}_4$ are straight lines in the plane $(d^{(2)}, \varepsilon^{(2)})$, since $L_1, L_2, \Omega_1, \Omega_2$ are independent of $d^{(2)}$ and $\varepsilon^{(2)}$, whereas $\sigma_1$ and $\sigma_2$ depend linearly on these parameters.

The effect of anomalous diffusion on the Turing--Hopf interaction is not structural but quantitative: the fractional exponents enter the problem through $k_c$, $\Omega_c$, and the normal form coefficients, thereby reshaping the bifurcation diagram and the relative size of the different dynamical regimes. 

In contrast to the pure Turing case, the dependence of the normal form coefficients on the fractional exponents is highly intricate. This prevents a clear separation of mechanisms and makes it difficult to identify simple trends in the variation of the instability boundaries. As a consequence, superdiffusion mainly affects the quantitative organization of the parameter space rather than the qualitative structure of the bifurcation scenario.

We now present a set of numerical simulations in order to validate the normal form analysis and to illustrate the different dynamical regimes arising in the neighborhood of the Turing--Hopf codimension-2 bifurcation.
\begin{figure*}
	\subfigure[]{\includegraphics[width=0.45\textwidth]{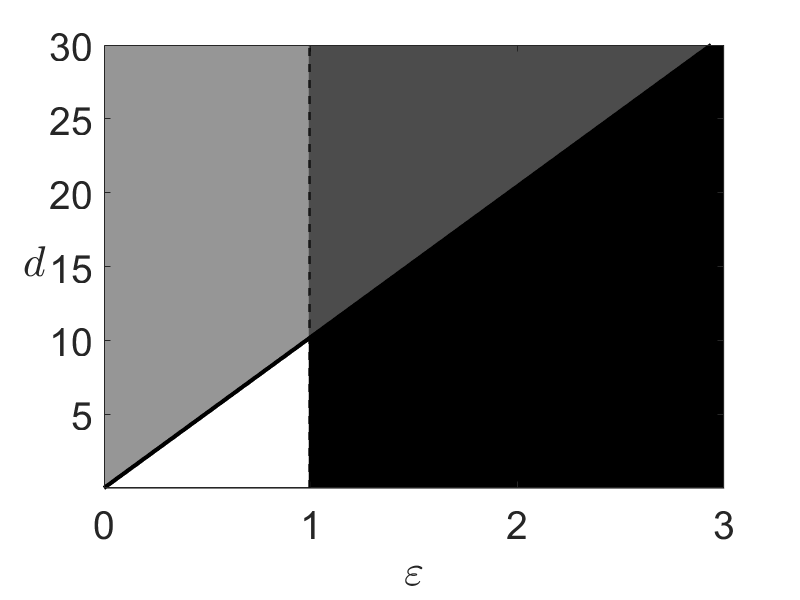}
		\label{TH_region}}
	\subfigure[]{\includegraphics[width=0.41\textwidth]{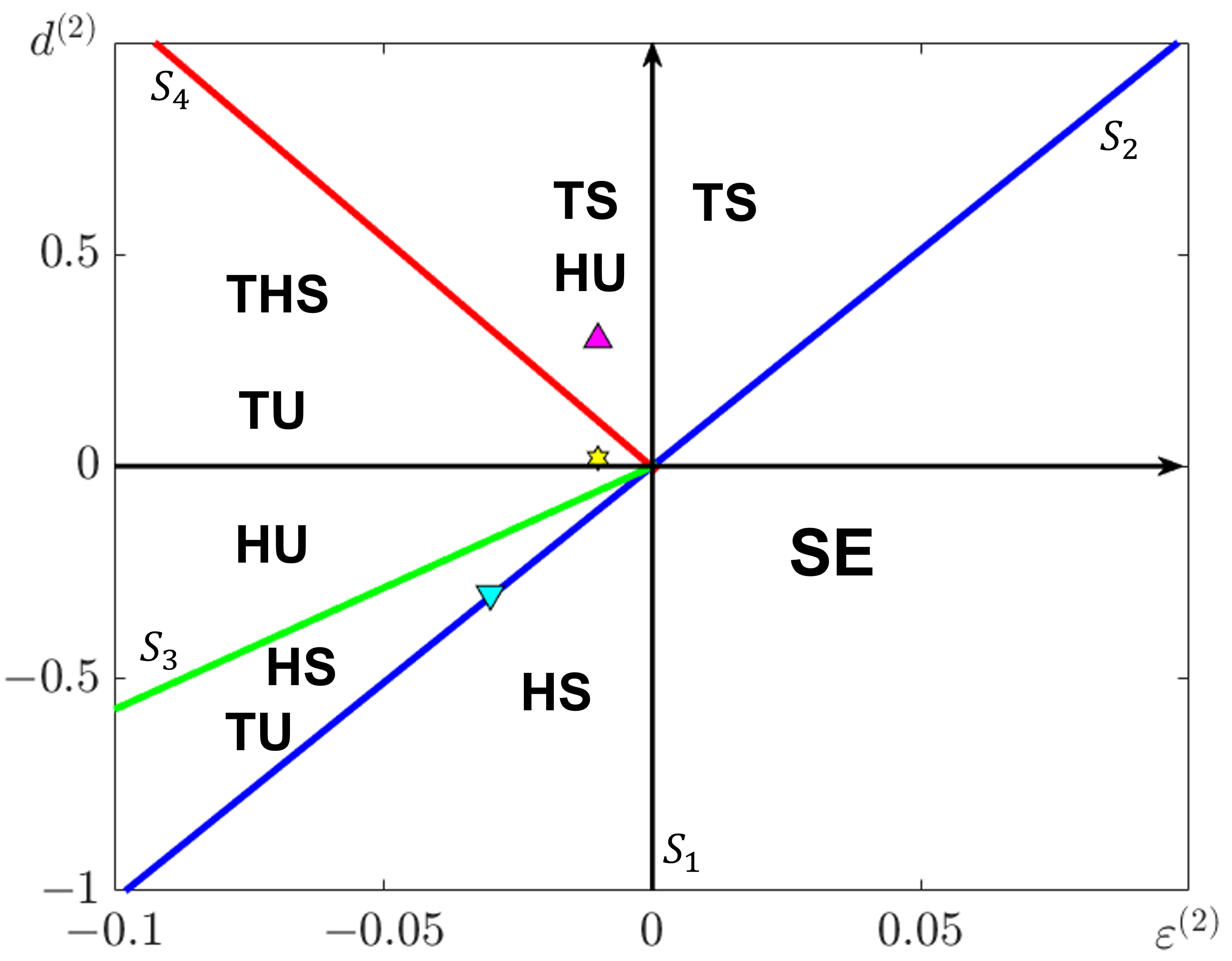}
		\label{TH_cond}}
	\caption{(a) Instability regions in the $(\varepsilon,d)$ parameter plane, where $\varepsilon$ and $d$ act as control parameters for Hopf and Turing instabilities, respectively. The white region corresponds to pure Hopf instability, the dark gray region to pure Turing instability, the light gray region to the Turing--Hopf interaction, and the black region to linear stability of the homogeneous equilibrium. Parameter values are $a=0.2$, $\beta=0.21$, $\gamma=3.1$, $\Gamma=13$, with $\bar{\alpha}=1.1$.
		(b) Bifurcation diagram in the $(\varepsilon^{(2)},d^{(2)})$ plane in a neighborhood of the codimension--2 Turing--Hopf point $(\varepsilon_c,d_c)=(0.9928,10.1398)$, obtained from the normal form \eqref{normal1}--\eqref{normal4}. The blue line represents the Turing bifurcation threshold (existence of solution $T$), while the vertical axis corresponds to the Hopf bifurcation threshold (existence of solution $H$). The green and red lines delimit the region where the mixed-mode solution $TH$ exists. These curves divide the parameter plane into regions characterized by the existence and stability of the equilibria $O$, $H$, $T$, and $TH$. The markers indicate the parameter values used in the numerical simulations shown in Fig.~\ref{TH_simul}.
	}
	\label{TH_fig}
\end{figure*}

In the first numerical experiment, the parameter values are chosen as specified in the caption of Fig.~\ref{TH_region}, where is represented the bifurcation structure associated with Turing and Hopf instabilities in the $(\varepsilon,d)$ parameter plane. The diagram identifies the regions of pure Turing and pure Hopf instability, together with the parameter domain in which both instabilities simultaneously arise.
For this parameter set the Turing and Hopf instability boundaries ntersect at the codimension--2 Turing--Hopf bifurcation point  $(\varepsilon_c, d_c) = (0.9928, 10.1398)$. 
Figure~\ref{TH_cond} reports the bifurcation lines $\mathcal{S}_1$--$\mathcal{S}_4$, predicted by the normal form analysis of system \eqref{normal1}--\eqref{normal4}, in the plane $(\varepsilon^{(2)}, d^{(2)})$. These lines partition the parameter plane into regions corresponding to the existence and stability of the equilibria $O$, $H$, $T$, and $TH$.

In particular, the blue line corresponds to the Turing bifurcation threshold (existence of solution $T$), while the vertical axis marks the Hopf bifurcation threshold (existence of solution $H$). The green and red lines delimit the region where the mixed-mode solution $TH$ exists and is stable. It is worth stressing that the position of the codimension--two point and the shape of the regions separated by the lines $\mathcal{S}_1$--$\mathcal{S}_4$ depend on the fractional diffusion exponents through the coefficients of the normal form. In particular, varying $\bar{\alpha}$ changes the relative size of the purely oscillatory, purely stationary, and mixed Turing--Hopf regions, even though the overall structure of the bifurcation diagram remains the same.

To validate the predictions of the normal form analysis, we perform direct numerical simulations of the full reaction--diffusion system \eqref{origsyst_1}--\eqref{origsyst_2}. The resulting spatio-temporal dynamics of the activator $u$ are displayed in Fig.~\ref{TH_simul}.

In Fig.~\ref{TH_simul}(a), the parameter values $(\varepsilon^{(2)}, d^{(2)}) = (-0.03, -0.3)$ correspond to the cyan downward triangle in Fig.~\ref{TH_cond}. In this region, only the Hopf equilibrium $H$ of the normal form exists and is stable. Accordingly, the solution exhibits spatially homogeneous temporal oscillations.

In Fig.~\ref{TH_simul}(b), the parameters $(\varepsilon^{(2)}, d^{(2)}) = (-0.01, 0.3)$ correspond to the magenta upward triangle in Fig.~\ref{TH_cond}. Here, only the Turing equilibrium $T$ is stable, and the solution converges to a stationary spatially periodic pattern.

Finally, in Fig.~\ref{TH_simul}(c), the parameters $(\varepsilon^{(2)}, d^{(2)}) = (-0.01, 0.02)$ correspond to the yellow star in Fig.~\ref{TH_cond}. In this region, the mixed-mode equilibrium $TH$ is stable. As predicted by the normal form, the solution displays a combined spatio-temporal behavior, characterized by spatial modulation together with temporal oscillations.

These numerical experiments confirm the validity of the weakly nonlinear analysis and show that the normal form correctly predicts not only the location of the bifurcation boundaries, but also the qualitative nature of the emerging dynamics in the different regions of parameter space.

\begin{figure*}
	\subfigure[]{\includegraphics[width=0.32\textwidth]{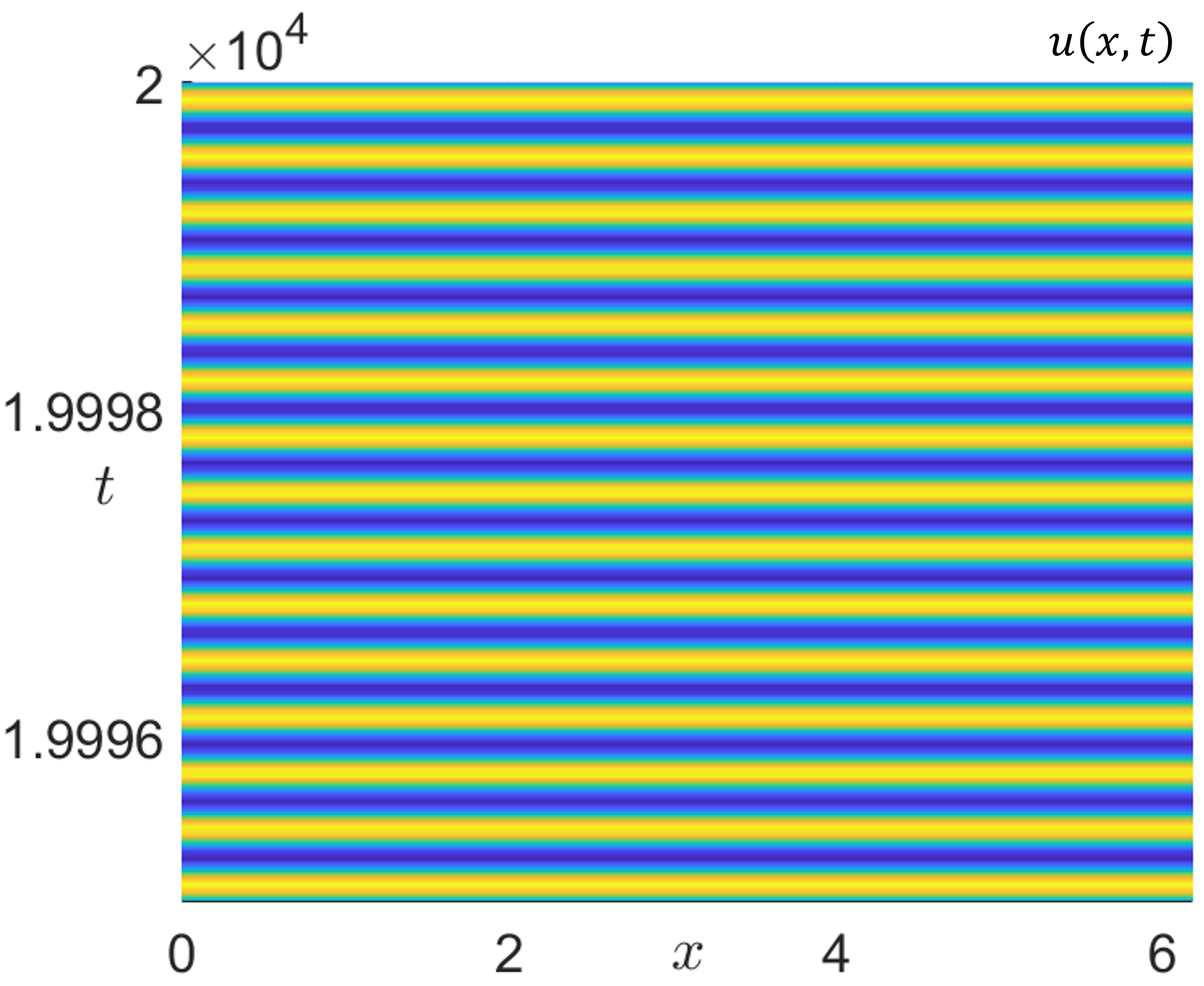}
		\label{u_H}}
	\subfigure[]{\includegraphics[width=0.32\textwidth]{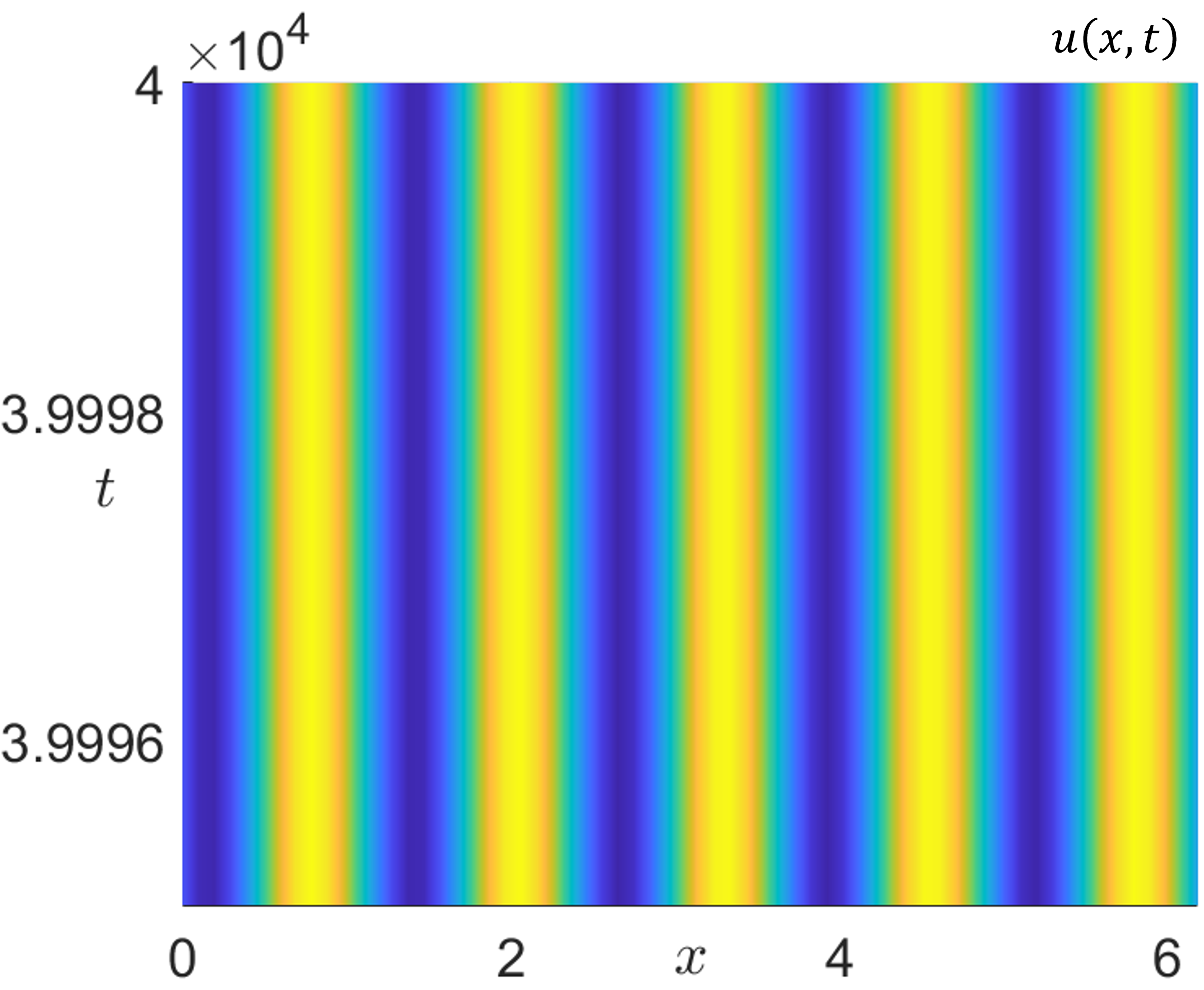}
		\label{u_T}}
	\subfigure[]{\includegraphics[width=0.32\textwidth]{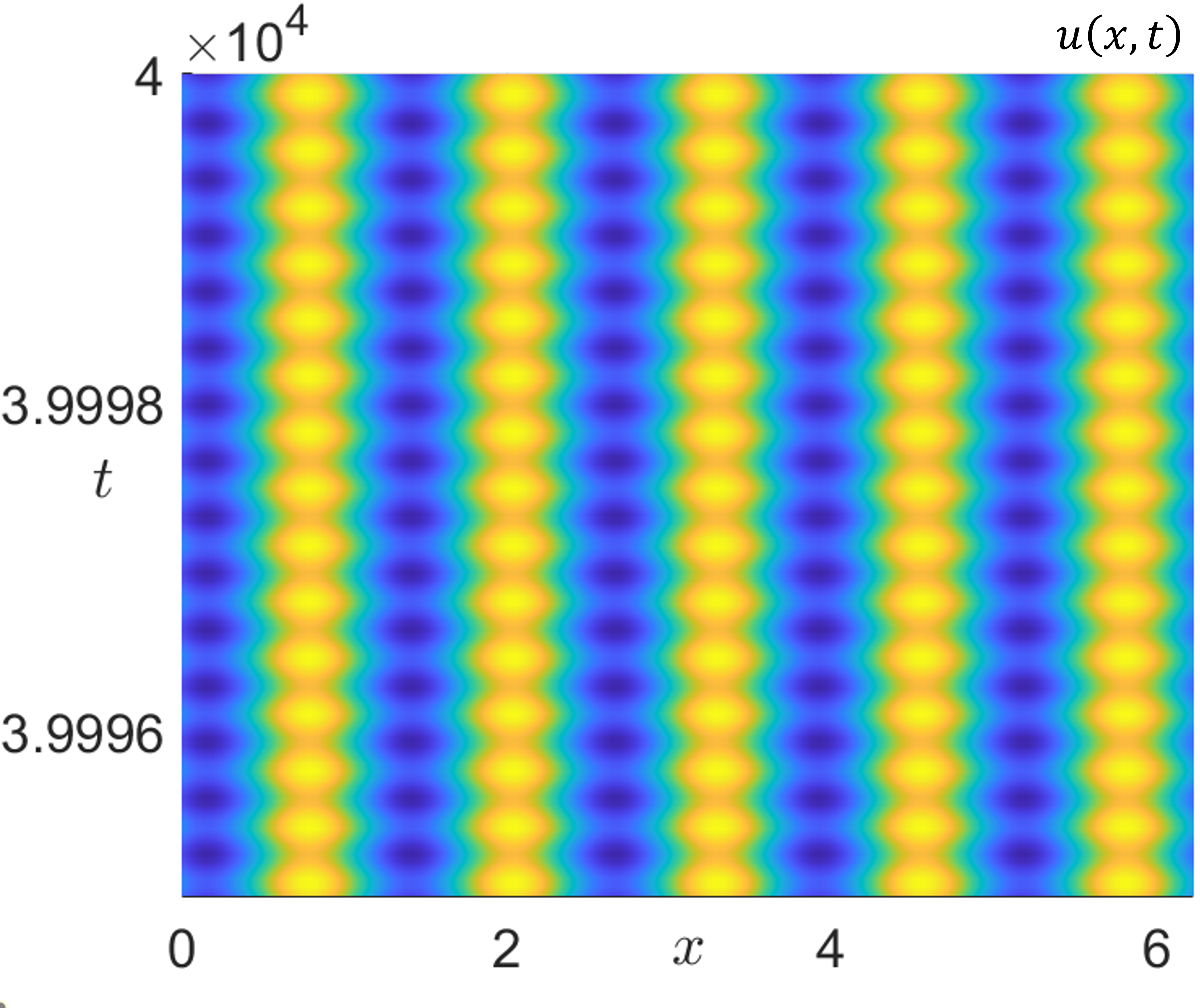}
		\label{u_TH}}
	\caption{Spatio-temporal solutions of system \eqref{origsyst_1}--\eqref{origsyst_2} in the vicinity of a codimension--2 Turing--Hopf bifurcation point. Parameter values are $a=0.2$, $\beta=0.21$, $\gamma=3.1$, $\Gamma=13$, with $\bar{\alpha}=1.1$.
		(a) Pure Hopf regime: the equilibrium $H$ of the normal form \eqref{normal1}--\eqref{normal4} exists and is stable. The parameters are $(\varepsilon^{(2)}, d^{(2)}) = (-0.03,-0.3)$, corresponding to the cyan downward triangle in Fig.~\ref{TH_cond}. The solution exhibits spatially homogeneous temporal oscillations.
		(b) Pure Turing regime: the equilibrium $T$ is stable. The parameters are $(\varepsilon^{(2)}, d^{(2)}) = (-0.01,0.3)$, corresponding to the magenta upward triangle in Fig.~\ref{TH_cond}. The solution converges to a stationary spatially periodic pattern.
		(c) Turing--Hopf regime: the mixed-mode equilibrium $TH$ is stable. The parameters are $(\varepsilon^{(2)}, d^{(2)}) = (-0.01,0.02)$, corresponding to the yellow star in Fig.~\ref{TH_cond}. The solution exhibits combined spatial modulation and temporal oscillations.
	}
	\label{TH_simul}
\end{figure*} 
This analysis shows that, also in the presence of anomalous diffusion, the FitzHugh--Nagumo system exhibits the classical interaction scenario between diffusion--driven and oscillatory instabilities. The superdiffusive transport does not destroy the Turing--Hopf mechanism, but reorganizes the parameter regions where stationary, oscillatory, and mixed spatio--temporal patterns occur. Together with the results of the previous sections, this completes the picture of pattern formation in the fractional FHN model, showing that anomalous diffusion reshapes not only stationary spatial structures but also the transition to spatio--temporal complexity.

\section{Conclusions}\label{conclusions}

In this work we have investigated diffusion-driven instabilities in a FitzHugh--Nagumo system with fractional transport and heterogeneous diffusion orders.

By deriving explicit analytical expressions for the instability threshold and the critical wavenumber, we have shown that the onset of instability is governed by the ratio of the diffusion exponents and by the kinetic parameters, rather than by their individual values. This property provides a natural framework for classifying instability regimes in anomalous reaction--diffusion systems.

When the diffusion orders differ, nonclassical instability mechanisms emerge. In particular, spatial instabilities may arise even in regimes where the activator diffuses faster than the inhibitor, as a consequence of the combined effect of anomalous scaling, diffusion rates and domain size. This behavior cannot be explained within the classical activator--inhibitor framework based on diffusion-rate separation.

The weakly nonlinear analysis shows that superdiffusion affects nonlinear saturation and promotes subcritical behavior. This enhances finite-amplitude effects and increases the sensitivity to perturbations.

The analysis of Turing--Hopf interactions shows that fractional transport shifts the boundaries between stationary and oscillatory instability regions.

From a broader perspective, our results highlight the importance of distinguishing diffusion rates from diffusion exponents when modeling transport in heterogeneous media. This distinction is relevant in biological, ecological and physical systems where long-range movements are present.

Future work will extend the present analysis to more general settings, including higher-dimensional domains and more realistic boundary and noise effects, in order to further clarify the role of anomalous transport in complex dynamical regimes. Studies of pattern formation on non-standard geometries, such as evolving or curved domains \cite{WBRGM2011,KEG2019,LFCD2022}, and analyses of Turing instability and pattern selection in two-dimensional reaction–diffusion systems \cite{GLS13,CCDT2022}, suggest complementary approaches that could be integrated with fractional transport effects. Additionally, coupling complex dynamical elements through fractional operators to investigate emergent spatio-temporal and potentially chaotic behavior presents an intriguing direction, as seen in coupled oscillator networks and related complex spreading dynamics \cite{ABCPPSV21,MARAM2024}.
It would be also interesting to explore whether the fractional transport mechanisms considered here can be derived from microscopic interaction rules via kinetic limits, as pursued in kinetic macroscopic studies of neural and chemotactic systems \cite{PST18,BGMT2025}.

\section{Acknowledgment}						
This work has been partially supported 
by the European Union – Next Generation EU: PRIN project
MUR PrinPNRR 2022, project code 
P202254HT8, Grant No. CUP B53D23027760001; by the PRIN project MUR PrinPNRR 2022, project code P2022Z7ZAJ, Grant No. CUP B53D23027930001 and by the project
MUR PRIN 2022-scorrimento project code 2022ALZH2K002, Grant No. CUP: B53C24006330006; by the Project {\it for single researcher} at University of Palermo, Engineering Department.
The authors also gratefully acknowledge the financial support of GNFM-INdAM and the FFR2024 grant of the University of Palermo.

	\bibliographystyle{plainnat}
	\bibliography{FHNBib}
	
\end{document}